\documentclass[11pt,reqno]{amsart}
\usepackage[utf8x]{inputenc}
\usepackage{ucs}
\usepackage[english]{babel}
\usepackage{amsmath,slashed}
\usepackage{amsfonts}
\usepackage{mathtools}
\usepackage{amssymb}
\usepackage[shortlabels]{enumitem}
\usepackage{bbm}
\usepackage{graphicx}
\usepackage[left=2cm,right=2cm,top=2cm,bottom=2cm]{geometry}

\usepackage{mathpazo} 

\newlength{\intwidth}

\usepackage{amssymb}
\usepackage[centertags]{amsmath}
\usepackage{verbatim}

\newcommand\om\omega

\newcommand{\foralll}{\forall\hspace{1mm}}
\newcommand{\existss}{\exists\hspace{1mm}}
\renewcommand\ell{{ l }}

\def\cN{\mathcal N}
\def\cC{\mathcal C}

\def\cB{\mathcal B}

\def\cS{\mathcal S}
\def\cA{\mathcal A}

\newcommand\vp\varphi
\newcommand\ka\kappa
\newcommand\te\theta
\newcommand\vka\varkappa

\newcommand\de\delta
\newcommand\La\Lambda
\newcommand\la\lambda
\newcommand\ga\gamma

\newcommand\Ga\Gamma

\newcommand\cI{\mathcal I}
\newcommand\cJ{\mathcal J}

\newcommand\cD{\mathcal{D}}
\newcommand{\triple}[1]{{\left\vert\kern-0.25ex\left\vert\kern-0.25ex\left\vert #1
        \right\vert\kern-0.25ex\right\vert\kern-0.25ex\right\vert}}

\usepackage[toc, page]{appendix}

\newcommand{\norm}[1]{\left\lVert#1\right\rVert}

\newcommand{\bR}{\mathbb{R}}
\newcommand{\bN}{\mathbb{N}}
\newcommand{\bE}{\mathbb{E}}
\newcommand{\bP}{\mathbb{P}}
\newcommand{\bZ}{\mathbb{Z}}

\makeatletter
\renewcommand*{\@fnsymbol}[1]{\ensuremath{\ifcase#1\or *\or \star\or ***\or
   \mathsection\or \mathparagraph\or \|\or **\or \dagger\dagger
   \or \ddagger\ddagger \else\@ctrerr\fi}}
\makeatother

\usepackage{fancyhdr}
\newcommand\ep\varepsilon

\newtheorem{theorem}{Theorem}[section]
\newtheorem{lemma}[theorem]{Lemma}

\newtheorem{proposition}[theorem]{Proposition}

\newtheorem{definition}[theorem]{Definition}

\theoremstyle{definition} 
\newtheorem{remark}[theorem]{Remark} 
\newtheorem{example}[theorem]{Example}

\numberwithin{equation}{section}

\renewcommand\leq\leqslant
\renewcommand\geq\geqslant

\usepackage[pdftex]{hyperref}
\hypersetup{
	colorlinks=true,
	linkcolor=blue,
	citecolor=blue,
	filecolor=blue,
	urlcolor=blue,
}

\title[On the probability of the CJT or the MoA]{On the probability of the Condorcet Jury Theorem or the Miracle of Aggregation}

\author{\'Alvaro Romaniega Sancho}
\address{Instituto de Ciencias Matem\'aticas, Consejo Superior de
 Investigaciones Cient\'\i ficas, 28049 Madrid, Spain}
\email{alvaro.romaniega@icmat.es}

\begin{document}
\maketitle

\begin{abstract}
	The Condorcet Jury Theorem or the Miracle of Aggregation are frequently invoked to ensure the competence of some aggregate decision-making processes. In this article we explore an estimation of the prior probability of  the thesis predicted by the theorem (if there are enough voters, majority rule is a competent decision procedure). We use tools from measure theory to conclude that, prima facie, it will fail almost surely. To update this prior either more evidence in favor of competence would be needed or a modification of the decision rule. Following the latter, we investigate how to obtain an almost sure competent information aggregation mechanism for almost any evidence on voter competence (including the less favorable ones). To do so, we substitute simple majority rule by weighted majority rule based on some weights correlated with epistemic rationality such that every voter is guaranteed a minimal weight equal to one.
 
\end{abstract}
\section{Introduction}\label{sec:intro}
\subsection{Epistemic processes of information aggregation: the Condorcet Jury Theorem and the Miracle of Aggregation}
Back in 1785 Condorcet published a result to show how voting could be useful to efficiently aggregate the private information of a group of agents. The result holds when we face a dichotomous choice between $A$ and $B$ which has a correct option, say $A$. For instance, the group of agents can be a jury which has to decide if the defendant is guilty in a criminal trial. Each agent is assumed to be more competent than a coin toss and their choices are assumed to be independent from each other. In this setting, Condorcet showed that if votes are aggregated using simple majority rule ($A$ wins if its number of votes is greater than the number of votes of $B$ with an odd number of voters), the probability of choosing the right option increases to one as the number of voters goes to infinity. Thus, \textit{we can efficiently aggregate information}: if the voters are slightly competent, we can produce an (almost) perfectly competent decision procedure, i.e., the probability of being right is as close as one as we want if there are enough voters. This is the Condorcet Jury Theorem (CJT). More precisely, the $i$-th voter is a random variable $X_i$ over $\{0,1\}$ with probability of choosing $A$ equal to $\bP\left(X_i=1\right)=p>1/2$ and, obviously, the probability of choosing $B$ is $\bP\left(X_i=0\right)=1-p<1/2$. Thus, if voters are i.i.d. random variables and the aggregation procedure is simple majority rule, i.e., $A$ is chosen if
\begin{equation*}
	\sum_{i=1}^n X_i>\frac{n}{2},
\end{equation*}
where $n$ is an odd number which represents the number of voters, then 
\begin{equation}\label{eq:CJT cond}
	\lim_{n\to\infty}\bP\left(\sum_{i=1}^n X_i>\frac{n}{2}\right)=1.
\end{equation}

If we consider a similar setting, but now most of the voters are no better than chance, i.e., $p=1/2$, except for a group of informed voters such that $p=1$, then \eqref{eq:CJT cond} holds too. The idea is simple, most of the voters behave like noise that cancels out (because $p=1/2$) and the informed group introduces a ``bias'' toward the right choice. More precisely, the result follows from the Weak Law of Large Numbers (WLLN). This second case is sometimes known in the literature as Wisdom of Crowds (WoC) or Miracle of Aggregation (MoA).

In general, we can ask whether \eqref{eq:CJT cond} happens for an arbitrary distribution of voter competence, i.e., for a general sequence of probabilities $\{p_i\}_{i=1}^\infty$ where $p_i\coloneqq\bP\left(X_i=1\right)$ (now voters are not necessarily identical). This is usually called the \textit{asymptotic CJT} for independent voters. The cases considered above were: 
\begin{itemize}
	\item Condorcet: $p_i=1/2+\varepsilon$ where $\varepsilon\in (0,1/2]$ $\forall i\in \bN$.
	\item MoA: given $n$ voters, $p_i=1/2$ for $(1-\varepsilon)n$ and $p_j=1$ for $\varepsilon n$ voters where $\varepsilon\in (0,1])$.\footnote{To simplify the exposition we have assumed that $\varepsilon n$ is an integer, but we should write $\varepsilon_n n\coloneqq\lfloor \varepsilon \cdot n\rfloor$, i.e., take the integer part.}
\end{itemize}
We will denote by $\cC_I$ the subset of $[0,1]^\infty$ or $[0,1]^\bN$, i.e., the space of sequences with elements in $[0,1]$ such that \eqref{eq:CJT cond} holds for independent voters. We will say that a sequence of probabilities satisfies the Condorcet Jury Property (CJP) if \eqref{eq:CJT cond} holds, i.e., the thesis of the CJT holds. This generalizes the case considered by Condorcet of $p_i=1/2+\varepsilon$. $\cC_I$ is an infinite set, i.e., there are an infinite number of sequences satisfying the CJP. This was proved in \cite{BP98}, where a complete characterization of the CJP was given.

Nevertheless, we can ask how ``large'' this set is compared with its complement (that is, the set of sequences where the CJP does not hold). Answering this will be one of the main topics of this paper. That is, we might want to know if the thesis of the theorem is likely to hold in order to evaluate the efficiency of some decision procedures or in order to know if it is ``legitimate'' to invoke this theorem. It does not seem plausible to measure the $p_i$ and check whether $(p_i)_{i=1}^\infty$ belongs to $\cC_I$ as this is impractical (or even impossible) for this simple model of the real world. A more sensible way could be to ask how likely is the event $\cC_I$ given that the main ``parameters'' of the problem, $(p_i)_{i=1}^\infty$, are partially unknown. Following a Bayesian approach, one should start computing the prior probability. One of the main difficulties of this is how to define and choose a probability distribution or measure in $[0,1]^\bN$. This will be analyzed in Section \ref{sec:apriori apc}.
\subsection{Weighted majority rule}
Let $w\coloneqq\left(w_i\in\bR\right)_{i=1}^\infty$ be some weights. Now, to each voter we associate the random variable
\begin{equation}
	X_i=\begin{cases*}
		+1 &\text{ if it votes A,}\\
		-1 &\text{ if it votes B.}
	\end{cases*}
\end{equation}
Weighted majority rule implies that the social choice function is $\text{sign}(X^w_n)$ being indifferent between the two if $X^w_n=0$ where
\begin{equation}\label{eq:weighted maj}
	X^w_n\coloneqq \sum_{i=1}^n w_i X_i.
\end{equation}
The larger the weight (\textit{ceteris paribus}), the greater the influence of the voter. The previous case of simple majority rule is recovered if $w_i=w_j$ $\foralll i,j$. Notice that if we assigned $X=0$ for the wrong option, the weights would be irrelevant in that case, so we have to consider the symmetric case $X\in\{-1,1\}$.

Weighted majority rules have been widely explored in the literature. For instance, in \cite{NP82} it is shown that, under some assumptions, weighted majority rule is the optimal decision rule for dichotomous choices and that the weights are given by
\begin{equation}\label{eq:optimal weights}
	w_i=\mathcal{W}(p_i)\coloneqq\log\left(\frac{p_i}{1-p_i}\right).
\end{equation}
Obviously, $\mathcal{W}:[0,1]\to\bR$ and, in particular, $\lim_{p\to 0}\mathcal{W}(p)=-\infty$ and $\lim_{p\to 1}\mathcal{W}(p)=\infty$. Also, $\mathcal{W}(p)<0$ for $p<1/2$. Some intuitions of this result were unveiled by Nobel-prize winner Lloyd Shapely and Bernard Grofman, \cite{SG84}. Considering the non-asymptotic CJT, suppose we have voters with competences (0.9, 0.9, 0.6, 0.6, 0.6). We have several options:
\begin{itemize}
	\item Under expert rule, $w_i=0$ only for $i=2,3,4,5$, 
	$$\bP\left(\sum_{i=1}^5 w_i X_i>0\right)=0.9\,.$$
	\item Under simple majority rule, $w_i=1$ for $i=1,2,3,4,5$, 
	$$\bP\left(\sum_{i=1}^5 w_i X_i>0\right)\approx 0.877\,,$$
	which improves the mean competence, but it is below expert rule.
	\item Under weighted majority rule, $w_i=1/3$ for $i=1,2$ and $w_i=1/9$ for $i=3,4,5$, 
	$$\bP\left(\sum_{i=1}^5 w_i X_i>0\right)\approx 0.927\,,$$
	which improves the previous results.
\end{itemize}
This result might be counterintuitive, since we are assigning nonzero weights to the less competent but, nevertheless, improving the total probability with respect to the expert rule case. This result is clearer if we note that these less competent members can break the tie if the two most competent individuals disagree. The use of weights \eqref{eq:optimal weights} was considered an important result by the aforementioned authors. They conclude \cite{SG84}:
\begin{quote}
	While the results of this essay seem particularly appropriate to analysis of the problem of 'information pooling', in which the task is to weigh the advice of 'experts' or to reconcile 'expert' and 'non-expert' conflicting opinion; we believe Theorem II [this is \eqref{eq:optimal weights}] to be of considerable general importance for democratic theory. 
\end{quote}
In that sense, we will also consider the CJT for a weighted majority rule and how probable it is now the CJT when weights are included. Nevertheless, we will explore a different kind of weights: they will be strictly positive, bounded from below and above and subject to some stochastic error. They will be of the form $w=w_d+\varepsilon$ where $w_d$ is a deterministic function depending on $p$ and $\varepsilon$ the error. See Figure \ref{fig:weights} for a comparison between $w_d$ and $\mathcal{W}$. 
\begin{figure}[h]
	\centering
	\includegraphics[width=.55\linewidth]{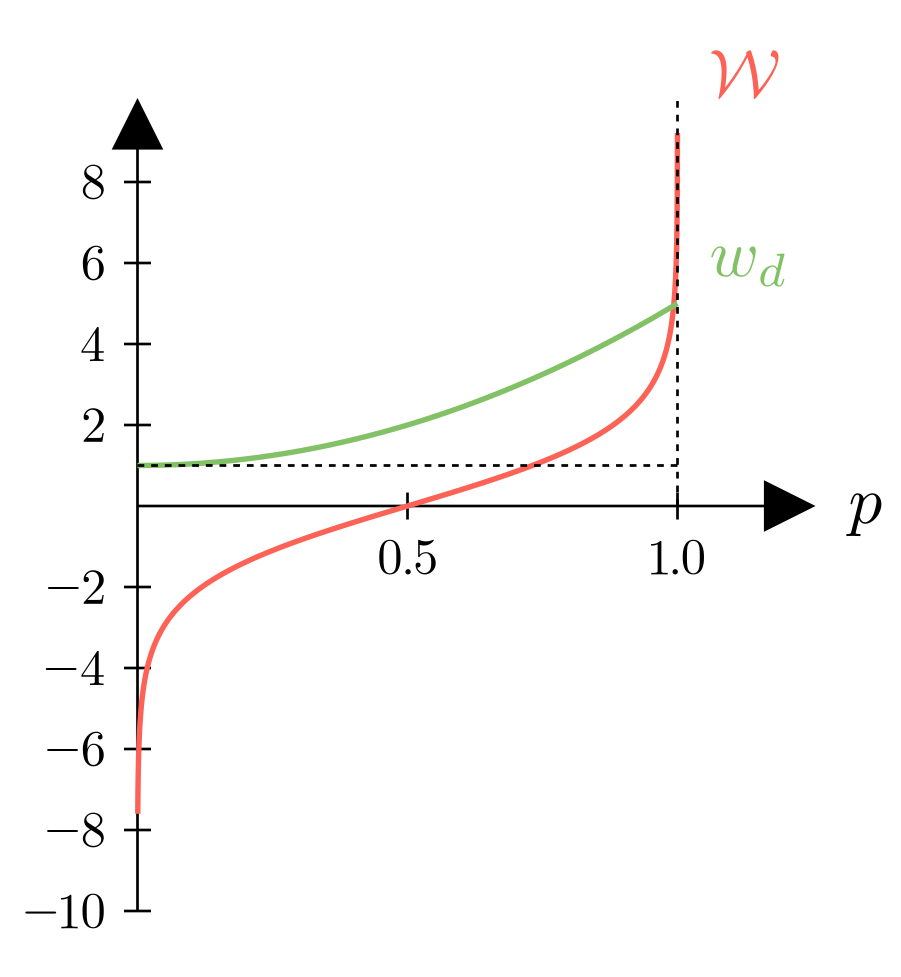}
	\caption[Comparison between optimal weights and bounded weights]{Comparison of optimal weights and bounded weights.}
	\label{fig:weights}
\end{figure}
We use $w_d$ instead of $\mathcal{W}$ because, although mathematically optimal, they can be problematic in a real life situation. In particular, as we said, $\mathcal{W}(p)<0$ for $p<1/2$. This has the effect of, for $p_i<1/2$,
\begin{equation*}
	\mathcal{W}(p_i) X_i=|\mathcal{W}(p_i)|\left(-X_i\right)=:|\mathcal{W}(p_i)|\tilde{X}_i=\mathcal{W}(1-p_i)\tilde{X}_i
\end{equation*}
and now $\bP\left(\tilde{X}_i=1\right)=\bP\left({X}_i=-1\right)=1-p_i>1/2$. This is equivalent to reversing  the outcome of the vote for these particular voters.
To avoid this, we will only consider weights in an interval of the form $[1,W]$ for some $W>1$, i.e., no voter loses, formally, its weight on the election. In the same manner, we will also assume there is an error of measurement, so weights are not perfectly correlated to competence. As we said, real weights, $w$, will be the sum of deterministic weights, $w_d(p)$, plus a random error, $\varepsilon$.
\subsection{Motivation of the paper and main results.} The Condorcet Jury Theorem or the Miracle of Aggregation are frequently invoked to ensure the competence of some aggregate decision-making processes. Furthermore, to the best of the author's knowledge, the current literature focuses on sufficient conditions (in different circumstances) to ensure the thesis of the theorem, but less attention has been paid to the applicability of the results. 

Our objective in this work is to set the framework for the study of the applicability of these important results. As directly checking the hypotheses of the theorem is unrealistic, we use a probabilistic approach with Bayesian grounds. Here, we study under which circumstances the thesis predicted by the theorem is likely to hold. Depending on our available evidence on voter competence, which will be measured by a bias in a second-order probability measure, the thesis of the theorem will happen almost surely or almost never. See Theorem \ref{th:unweig and indep} and Theorem \ref{th:unweig, indep and pos bias} for details. As we will see in these theorems, the opposite of the CJP could occur almost surely, i.e., majority rule chooses the wrong option a.s. Therefore, this gives another reason to study the applicability in order to  ensure that we are not in this situation.

Furthermore, we also apply this framework in the case of weighted majority rule with stochastic (or noisy) weights. It is concluded that these stochastic weights can fix almost any voter profile of incompetence, see Theorem \ref{thm:measure weighted CJT}.
\subsection{Organization of the paper}

The paper is organized as follows. In Section \ref{sec:notation} we introduce the notation and some definitions that will be used in the rest of the paper. In Section \ref{sec:apriori apc} we introduce the topic of the applicability of the CJT and present one of the main results, it fails almost surely. In Appendix \ref{ap:extension} we present a (more technical) generalization of the result of the preceding section. In Section \ref{sec:weighted CJT} we analyze the applicability when weighted majority rule is used instead of simple majority rule. An important part, Appendix \ref{ap:psy and phil}, is devoted to its practical implementation. Section \ref{sec:conclusion} gives an end to this paper offering some concluding remarks.

The main proofs are relegated to the appendices. In Appendix \ref{ap:proof unweig} and Appendix \ref{ap:extension} we prove the theorems of the unweighted situation.  Finally, in Appendix \ref{ap:proof weight} we prove the theorem concerning weighted majority rule.
\section{Notation and some definitions}\label{sec:notation}
The space of sequences with elements in $[0,1]$ will be denoted by $[0,1]^\bN$ where $\bN\coloneqq\bZ_{>0}=\{1,2,3,\ldots\}$. Odd numbers will be represented by $\mathbb{O}$. The cardinality of a set $A$ will be represented by $|A|$. The (uncentered) moments of a measure (this can be understood as a probability distribution) $\nu_0$ will be denoted by:
\begin{equation}\label{eq:moments}
	m^i\coloneqq \int_{[0,1]} x^i~ d\nu_0 (x)\le 1.
\end{equation}
In particular, $m\coloneqq m^1=b+\frac{1}{2}$ following Definition \ref{def:biased measure}. Also, we will denote $\epsilon_1\coloneqq \nu_0(\{1\})$. More generally, we define $\epsilon_{1-\varepsilon_0,1}\coloneqq \nu_0\left([1-\varepsilon_0,1]\right)$.   

We denote by $(\Omega,\mathbb{P})$ an abstract probability space where every random variable $X_i$ is defined. Given two measures $\nu,~\nu'$ we will say that $\nu$ is absolutely continuous with respect to $\nu'$ and write $\nu\ll\nu'$ if for every Borel set $A$, $\nu'(A)=0$ implies $\nu(A)=0$. We will write its Radon–Nikodym derivative as $\frac{d\nu}{d\nu'}$. This can be seen as the probability density of a distribution given by $\nu$. If there is a $C>0$ such that $a\le C b$ we write $a\lesssim b$.

\section{On the {a priori} applicability of those results}\label{sec:apriori apc}
\subsection{Preliminary example.}
We are going to use some concepts of measure theory, a well-established mathematical branch.
In the segment $[0,1]$ we ``measure'' a subset $A$ using the length or, in general, the Lebesgue measure $\lambda$. Thus, if $[a,b]\subset[0,1]$, its measure is $\lambda\left([a,b]\right)=b-a$, its length. Let us now add another dimension. In the square $[0,1]^2$ we ``measure'' a subset $A$ using the measure on $\bR^2$ $\mu=\lambda\times\lambda=\prod_{n=1}^{2}\lambda$, that is, $\mu$ is the area. For instance, if $A$ is the smaller square $A=[0,\frac12]^2$, then
\begin{equation*}
	\mu(A)=\lambda\left(\left[0,1/2\right]\right)\lambda\left(\left[0,1/2\right]\right)=\left(\frac12\right)^2.
\end{equation*}
We could use this measure for some probability events as follows. For instance,  if we define $X\coloneqq X_1+X_2$, $X_i\sim Bernoulli(p_i)$ and $p_i$ are unknown, then $\{\bE[X]<1\}$ has measure $1/2$ w.r.t. $\mu$. Indeed, $\bE[X]=p_1+p_2<1$ and by basic geometry
$$
\mu\{(x,y)\in[0,1]^2~/~x+y<1\}=\frac12.
$$
In the same fashion, we can see that the measure of $\{\bE[X]\le 2\}$ is 1 as $p_1,p_2\le 1$ and, similarly, the measure of $\{\bE[X]\ge 2\}$ is 0. We then say that the event $\{\bE[X]\le 2\}$ happens almost surely or $\mu-$almost surely and $\{\bE[X]\ge 2\}$ does not happen $\mu-$almost surely ($\mu$-a.s.). In this setting, we can think of $\mu$ as a ``\textit{meta-probability measure}'' or a second-order probability measure, it assigns probabilities (or measures) to some events of the parameters of the probability distributions of some random variables of our interest. Thus, we have two different probability spaces\footnote{The $\sigma$-algebra will be the standard one in each case and thus it will be implicitly assumed.}:
\begin{itemize}
	\item	\textit{Standard probability space} $(\Omega,\bP)$: the space (with its respective probability measure) depending on some parameters (fixed) where the problem is formulated. In our previous example it was given by the random variable $X:\Omega\to\bR$ with $\Omega$ the sample space and where the distribution of $X\equiv X_{p_1,p_2}$ depends on the fixed parameters $p_1$ and $p_2$. That is, for a measurable set $A$
	\begin{equation*}
		\bP\left(X\in A\right)=P_{X}(A,p_1,p_2).
	\end{equation*}
	\item 	\textit{Meta-probability space} $(\mathfrak{P},\mu)$: the space $\mathfrak{P}$ (with its respective probability measure $\mu$) of the parameters of the previous random variable. In our previous example it was given by $\mathfrak{P}=[0,1]^2$ and $\mu=\lambda\times\lambda$, the standard Lebesgue measure on the square $[0,1]^2$.
\end{itemize}
Now, it might be the case that we do not know the value of $p_1$ and $p_2$ but nevertheless want to know how ``likely'' will be that, for instance, $\bE[X]<1$. As we saw above, this is a problem involving the two probability spaces:
\begin{itemize}
	\item Standard probability space, $\bE[X_{p_1,p_2}]=\int_{\Omega}X_{p_1,p_2}(\omega)d\bP(\omega)=p_1+p_2$.
	\item Meta-probability space, $\mu\left(\{(p_1,p_2)\in[0,1]^2~|~p_1+p_2<1\}\right)=1/2$.
\end{itemize}
Notice also that if we chose a different $\mu$, the associated measure of each event would probably change, i.e., we have to choose the measure on $\mathfrak{P}$. From a Bayesian point of view, if we want to consider the prior probability, it can be assumed that this measure is not ``biased'' in any particular direction. That is, if we have no particular evidence to assume the contrary or prior to collect any evidence, it seems reasonable to impose that, for instance,
\begin{equation*}
	\mu(\{p_1\in[0,1/2)\})=\mu(\{p_1\in(1/2,1]\}).
\end{equation*}
\subsection{The CJT and measures on $[0,1]^\bN$}
For the CJT, given our previous definitions, we have:
\begin{itemize}
	\item	\textit{Standard probability space} is also denoted by $(\Omega,\bP)$ where the main random variables involved are $\frac{1}{n}\sum_{i=1}^n X_i$ for $n\in\bN$ and the event of our interest is given by \eqref{eq:CJT cond}.
	\item 	\textit{Meta-probability space} $(\mathfrak{P},\mu)$ equals $\left([0,1]^\bN,\mu\right)$. Here we are not interested in measures on $[0,1]^2$, but on $[0,1]^\infty$ or $[0,1]^\bN$, i.e., the space of sequences with elements in $[0,1]$, as $p_n\in[0,1]$ and the parameters of the problem are $\{p_n\}_{n=1}^\infty$.
\end{itemize}
We now turn into the problem of finding $\mu$ (or, more precisely, a set of $\mu$). A natural measure to consider is
\begin{equation}\label{eq:inf lebesgue}
	\mu=\prod_{n=1}^{\infty}\lambda,
\end{equation}
which is the generalization of the measure on $\bR^2$ considered above. It is well-defined by Kolmogorov's Extension Theorem. This measure has the property of being centered in the sense that the mean value (first moment) of $\lambda$ is
\begin{equation}
	\int_{[0,1]} x ~ d\lambda (x)=\frac12.
\end{equation}
However, we are going to consider more general ``centered'' measures than the one in \eqref{eq:inf lebesgue}, i.e., a larger class. Before the precise definition, we need to introduce the concept of distances and divergences of probability measures, say $d$. These objects tell us, in a sense to be precise in Appendix \ref{ap:proof unweig} , how different two distinct $\mu$ and $\mu'$ assign measures to an arbitrary set $A$. If $d(\mu,\mu')=0$, the measures are identical and if $d$ increases, so does the discrepancy for some sets. There are several ways of doing so, but two of the most important examples are the total variation distance (the statistical distance) and the Kullback--Leibler divergence (associated to the Shannon--Boltzmann entropy). In fact, we are going to consider a larger set, that will be denoted by $\cD$ and which will be defined precisely in  Definition \ref{rem:def cD}. To ease the exposition here, it can be understood that $d$ below is either the total variation distance or the Kullback--Leibler divergence. We are ready to define the concept of centered measures.
\begin{definition}\label{def:centered measure}
	A probability measure $\mu=\prod_{n=1}^{\infty}\nu_n$ on $[0,1]^\bN$ will be centered if there exists a probability measure on $[0,1]$, $\nu_0$, such that $\nu_n\ll\nu_0~\foralll n\ge 1$ (see Section \textnormal{\ref{sec:notation}} for notation),
	\begin{equation}\label{eq:centered}
		\int_{[0,1]} x~ d\nu_0 (x)=\frac12
	\end{equation}
	and
	\begin{equation}\label{eq:finite distance}
		\sum_{n=1}^\infty d(\nu_n,\nu_0)<\infty,
	\end{equation}	
with $d\in\cD$.
\end{definition}
\begin{example}\label{ex:cent}
	The case considered in \eqref{eq:inf lebesgue} corresponds to the case $\nu_0=\lambda$ and $\nu_0=\nu_n$ $\foralll n$ positive integers, so $d(\nu_0,\nu_n)=d(\nu_0,\nu_0)=0$	(by definition of distance and divergence, see Definition \ref{def:Div} and \ref{def:distance}) and then,
	\begin{equation*}
			\sum_{n=1}^\infty d(\nu_n,\nu_0)=0<\infty.
	\end{equation*}
\end{example}
The idea is simple, the measure $\mu$ is not too far (in the sense that the sum of distances or divergences does not go to infinity) from a product measure $\prod_{n=1}^{\infty} \nu_0$ of identical measures on $[0,1]$ and these measures have mean 1/2. This generalizes \eqref{eq:inf lebesgue} in two ways. First, the measures of the product are not necessarily identical. We allow the measure to be a ``perturbation'' of $\mu_0$. Second, the measure $\nu_0$ is not necessarily the Lebesgue measure, but a measure with mean $1/2$, i.e., we only need this measure to have the same first moment as the Lebesgue measure on $[0,1]$. For instance, we can have atomic measures, i.e., $\nu_0(\{x\})>0$ for some $x$. This is not allowed in the standard Lebesgue measure, as every single point has measure zero. In particular, as we said in Section \ref{sec:notation} we will define $\epsilon_1\coloneqq \nu_0(\{1\})$, that is, there is a probability $\epsilon_1$ such that each voter is going to vote the correct option almost surely as in the MoA. More generally, we define $\epsilon_{1-\varepsilon_0,1}\coloneqq \nu_0\left([1-\varepsilon_0,1]\right)$.   

With these measures, the CJP will not hold almost surely. It is important to note that as we have a complete characterization, we are not saying that the hypothesis of the theorem (CJT) will not hold, but \textit{that the thesis} (CJP) will not hold. The latter implies the former but the former implies the latter only if the conditions are necessary too. More precisely:
\begin{theorem}\label{th:unweig and indep}  Almost surely independent Condorcet Jury Theorem does not hold for a centered measure $\mu$, that is:
	\begin{equation}
		\mu(\mathcal{C}_I)=0.
	\end{equation}
\end{theorem}
That is, no matter which measure we choose (with the reasonable condition of Definition \ref{def:centered measure}), it will assign probability zero to the CJP.
\begin{remark}
	We should distinguish two concepts, impossible events and probability (or measure) zero events. The first are associated with the empty set $\emptyset$ and the second with a set of probability or measure zero, i.e., a null set. For instance, take a uniform random variable $X:\Omega\to [0,1]$ over $[0,1]$ and let $x_0\coloneqq\pi/4\in[0,1]$. $X$ will always (surely) give a number between 0 and 1. Thus,
	$\{X>1\}\coloneqq \{\omega~/~X(\omega)>1\}$ is an impossible event.  $\{X=x_0\}$ is not impossible (some number must be chosen and it also had probability zero) but  will not happen almost surely (the probability is zero). This implies that if we run the variable a large number of times $n$, $\frac{\text{number of times }X=x_0}{n}\to 0$ as $n\to \infty$ (probabilistic application to frequencies).
\end{remark}
Again, the idea is that whatever measure $\mu$ we choose with the condition that $\mu$ is centered or not ``biased'', the CJP will not hold almost surely. One could think that this was somehow expected as soon as we chose $m^1=1/2$: we are choosing $m^1=1/2$ but, as already Condorcet noticed, you need a probability greater than $1/2$. Thus, one does not expect the CJT to hold. Hence, the theorem is more or less trivial.  Nevertheless, this intuition would be incorrect, as it would be confusing the two probability spaces. Indeed, the first and second $1/2$ belong to two different spaces
\begin{itemize}
	\item $\bE[X_i]=\int_{\Omega}X_i(\omega)d\bP(\omega)=p_i$ and this must be greater than 1/2 in the standard CJT (where $p_i=p~\forall i\in\bN$),
	\item $m^1=\int_{[0,1]} x~ d\nu_0 (x)=\frac12$.
\end{itemize}
The confusion is obvious if we consider the homogeneous case $p_i=p~\forall i\in\bN$, the original Condorcet's theorem. Then, we would have $\mathfrak{P}=[0,1]$. Imposing that $\mu=\nu_0$ is centered around 1/2, i.e., $\int_{[0,1]} x~ d\mu (x)=\frac12$ would not imply that the CJT fails almost surely. In fact, it would have probability $\mu((1/2,1])>0$ unless $\mu=\delta_{1/2}$, the Dirac measure at $1/2$. For instance, if $\mu=\lambda$, then the CJT would have probability $1/2$.

A more subtler argument would say that, on average, probabilities are approximately (and asymptotically) $1/2$ because $m^1=1/2$. More precisely, $\frac{1}{n}\sum_{i=1}^{n}p_i\to \frac{1}{2}$ a.s. as $n\to\infty$. Thus, again, we cannot expect the CJT to hold in that situation because the probabilities, on average, should be greater than $1/2$. But, this intuition is incorrect too. In the following two examples we are going to construct an uncountable set of sequences where the CJT holds but $\frac{1}{n}\sum_{i=1}^{n}p_i\to \frac{1}{2}$ as $n\to\infty$, i.e., probabilities are on average 1/2.
\begin{example} Let 
	$$\cC_1\coloneqq \left\{(p_n)_{n=1}^\infty\in[0,1]^\bN~|~p_n=\frac12+\varepsilon_n,~\varepsilon_n\in \left[0,\frac12\right],~ \lim_{n\to\infty}\frac{1}{n}\sum_{i=1}^{n}\varepsilon_i=0,~ \lim_{n\to\infty}\frac{1}{\sqrt{n}}\sum_{i=1}^{n}\varepsilon_i=\infty\right\}.$$
	Then, $\cC_1\subset\cC_I$, i.e., the sequences in $\cC_1$ satisfy the CJP. Indeed, let $\bar{X}_n\coloneqq\frac{1}{n}\sum_{i=1}^{n}X_n$ as in Section \ref{sec:intro}. By definition,
	\begin{equation*}
		\bE\left(X_i\right)=p_i,~~\text{Var}\left(X_i\right)=p_iq_i,
	\end{equation*}
	where $q_i\coloneqq 1-p_i$, then
	\begin{align}\label{eq:ex1}
		\bP\left(\bar{X}_n\le 1/2\right)&=\bP\left(\bar{X}_n-\bE\left(\bar{X}_n\right)\le 1/2-\bE\left(\bar{X}_n\right)\right)\le \bP\left(|\bar{X}_n-\bE\left(\bar{X}_n\right)|\ge \bE\left(\bar{X}_n-1/2\right)\right)\le\nonumber\\
		&\le \frac{\text{Var}(\bar{X}_n)}{ \left(\bE\left(\bar{X}_n\right)-\frac12\right)^2}=\frac{\sum_{i=1}^n (1/4-\varepsilon_i^2)}{\left(\sum_{i=1}^n\varepsilon_i\right)^2}=\frac{\frac{1}{n}\sum_{i=1}^n (1/4-\varepsilon_i^2)}{\left(\frac{1}{\sqrt{n}}\sum_{i=1}^n\varepsilon_i\right)^2}\to 0
	\end{align}
	as $n\to\infty$ by Chebyshev's inequality, which can be applied because, by hypothesis, $\sum_{i=1}^n\varepsilon_i>0$ if $n$ is large enough. We have also used that
	\begin{equation*}
		\lim_{n\to\infty}\frac{1}{n}\sum_{i=1}^n (1/4-\varepsilon_i^2)=\frac14-\lim_{n\to\infty}\frac{1}{n}\sum_{i=1}^n \varepsilon_i^2=\frac14
	\end{equation*}
	because $\varepsilon_i^2\le \varepsilon_i$. Thus, $\bP\left(\bar{X}_n> 1/2\right)\to 1$ as $n\to\infty$, i.e., \eqref{eq:CJT cond}. But note that:
	\begin{equation*}
		\frac{1}{n}\sum_{i=1}^{n}p_i=\frac{1}{2}+\frac{1}{n}\sum_{i=1}^{n}\varepsilon_i\to \frac12.
	\end{equation*}
Thus, $p_i$ are, on average, $1/2$ but nevertheless the CJP holds.
\end{example}
\begin{remark} We can easily construct elements of this set as follows. Define $\varepsilon_i\coloneqq \max\{i^{\alpha},1/2\}$. Then, by the Euler--Maclaurin formula,
	\begin{equation*}
		H_n^{(-\alpha)}\coloneqq\sum _{i=1}^{n} i^{\alpha }=\frac{n^{\alpha +1}-1}{\alpha +1}+O\left(n^{\alpha}\right).
	\end{equation*}
Thus, it is enough if we take $\alpha\in(-1/2,0)$.	$H_n^{(-\alpha)}$ is the generalized harmonic number.	
\end{remark}
Now we present a second example of sets where, on average, the probabilities are 1/2 but the CJP holds. The idea behind the construction is completely different. It will also illustrate an important fact, being an element of $\cC_I$ \textit{does not necessarily depend only on the tail of the sequence.}
\begin{example} Let us fix an $m\in \bN$ greater than 1.
	Consider the sequence of $({p}_i)_{i=1}^\infty=({p}_1,\ldots,{p}_m,1,1,0,1,0,1,\ldots)$. In what follows we assume $p_i\in\{0,1\}$. In this setting where the probability is either 0 or 1, the CJT holds trivially iff $|\{i:1\le i\le n \text{ and } p_i=1\}|>|\{i:1\le i\le n \text{ and } p_i=0\}|$ for every $n$ large enough. Thus, this is equivalent to:
	\begin{equation}\label{eq:ex2}
		S_n\coloneqq|\{i:1\le i\le n \text{ and } p_i=1\}|=\sum_{i=1}^n {p}_i>\frac{n}{2}\quad \foralll n>n_0,
	\end{equation}	
both $n,n_0\in\mathbb{O}$. 	If ${p}_i=0~\foralll i\in\{1,\ldots,m\}$, then for $n=2k+1$:
	\begin{equation*}
		\frac{S_{n+m}}{n+m}-\frac{1}{2}=\frac{1+k}{2k+1+m}-\frac12=\frac{1-m}{2(m+2k+1)}<0\quad \foralll n\in\mathbb{O}.
	\end{equation*}
	But, on the other hand, if ${p}_i=1~\foralll i\in\{1,\ldots,m\}$, then:
	\begin{equation*}
		\frac{S_{n+m}}{n+m}-\frac{1}{2}=\frac{1+k+m}{2k+1+m}-\frac12=\frac{1+m}{2(m+2k+1)}>0\quad \foralll n\in\mathbb{O}.
	\end{equation*}
In general, if there are $m'$ $p_i=1$ for $1\le i\le m$ and $2m'+1>m$ then
\begin{equation*}
	\frac{S_{n+m}}{n+m}-\frac{1}{2}>0\quad \foralll n\in\mathbb{O}.
\end{equation*}
But, note again that 
\begin{equation*}
	\frac{1}{n}\sum_{i=1}^{n}p_i\to \frac12.
\end{equation*}
This latter set defines $\cC_2$.
\end{example}
Hence, as promised, $\cC_1\cup\cC_2$ is an uncountable set of sequences where the CJT holds but $\frac{1}{n}\sum_{i=1}^{n}p_i\to \frac{1}{2}$ as $n\to\infty$, i.e., probabilities are on average 1/2. These remarks warn us that the proof of Theorem \ref{th:unweig and indep} cannot rely on those intuitions and must use different ideas. This will be done in Appendix \ref{ap:proof unweig}. The basic idea is that we can only have the CJP if the sequences satisfy something similar to \eqref{eq:ex1} or \eqref{eq:ex2} of the previous examples. But these conditions are too restrictive and, thus, this set will have measure zero for the measures under consideration.

In the same manner, it could be argued that the hypotheses of the MoA are not satisfied. If the event that a proportion of voters is informed is quite rare for those measures, the MoA cannot be expected. Nevertheless, this condition of the MoA is satisfied in the following sense. First, recall that in Section \ref{sec:notation} we defined $\epsilon_1\coloneqq \nu_0(\{1\})$ and $\epsilon_{1-\varepsilon_0,1}\coloneqq \nu_0\left([1-\varepsilon_0,1]\right)$. This measures the probability that an individual voter is well-informed ($\epsilon_1$) and almost well-informed ($\epsilon_{1-\varepsilon_0,1}$, the probability of choosing the correct option is greater than $1-\varepsilon_0$ for some $\varepsilon_0$ generally small). Then, we have the following result (the cardinality of a set $A$ will be represented by $|A|$, Section \ref{sec:notation}):
\begin{proposition}\label{prop:MoA cond}
	Let $\mu$ a centered measure, $0\le\varepsilon_0<1/2$,	
	$0<\varepsilon<\epsilon_{1-\varepsilon_0,1}$ and $\delta>0$ as small as we want. Then,
	\begin{equation*}
		\lim_{n\to\infty}\frac{1}{n}|\{1\le i\le n~/~p_i\in[1-\varepsilon_0,1]\}|>\varepsilon~~ \textnormal{ a.s. }
	\end{equation*} 
	In particular, if $\varepsilon_0=0$ then the same holds with $p_i=1$ and $\epsilon_{1-\varepsilon_0,1}=\epsilon_1$.
\end{proposition}
This proposition tells us that, almost surely, at least a proportion $\varepsilon>0$ of voters is well-informed or almost well-informed if $n$ is large enough. These voters will vote for the correct option with probability greater than $ 1-\varepsilon_0$ with $\varepsilon_0$ as small as we want or even zero. Nevertheless, the thesis of the MoA is not satisfied.
\begin{remark}
	We will actually prove a ``finite'' version of the statement above. In particular, for every $\delta>0$ $\existss N\in\bN$ such that
	\begin{equation*}
		\mu_0\left(\left|\{1\le i\le n~/~p_i\in[1-\varepsilon_0,1]\}\right|>\varepsilon n\right)>1-\delta~~\foralll n\ge N,
	\end{equation*}  
	where $\mu_0=\prod_{n=1}^{\infty}\nu_0$. This statement means that a proportion $\varepsilon>0$ of voters is well-informed or almost well-informed will be reached if the population $n$ is greater than a (finite and uniform) $N$ with probability as close to one as we want.
	 To compare, the statement above is equivalent to:
	\begin{equation*}
		\mu\left(\lim_{n\to\infty}n^{-1}|\{1\le i\le n~/~p_i\in[1-\varepsilon_0,1]\}|>\varepsilon\right)=1.
	\end{equation*} 

\end{remark}

\subsection{On the election of $\mu$ and the prior probability}
To derive $\mu(\cC_I)=0$, the centered condition of Definition \ref{def:centered measure} can be relaxed somehow (although this condition is essential to calculate the a priori probability as we will see below). We could define in the same fashion as in Definition \ref{def:centered measure}:
\begin{definition}\label{def:biased measure}
	A probability measure $\mu=\prod_{n=1}^{\infty}\nu_n$ on $[0,1]^\bN$ will be $b$-biased for $b\in[-\frac12,\frac12]$ if there exists a probability measure on $[0,1]$, $\nu_0$, such that $\nu_n\ll\nu_0~\foralll n\ge 1$ (see Section \textnormal{\ref{sec:notation}} for notation), 
	\begin{equation}
		\int_{[0,1]} x~ d\nu_0 (x)=\frac12+b
	\end{equation}
	and
	\begin{equation}\label{eq:dist cond}
		\sum_{n=1}^\infty d(\nu_n,\nu_0)<\infty.
	\end{equation}	
	with $d\in\cD$,
\end{definition}
\begin{example}
	This example is a generalization of the Example \ref{ex:cent}. 	For instance, consider a ``biased'' measure
	\begin{equation}\label{eq:biased mu}
		\mu_{b_0}=\prod_{n=1}^{\infty}\lambda_{b_0},
	\end{equation}
	where the Radon-Nikodym derivative  (this is, its probability density function $\rho_{b_0}$) is given by
	$$
	\frac{d\lambda_{b_0}}{d\lambda}(x)=\rho_{b_0}(x)=(1-b_0/2)+b_0 x,
	$$
	with $b_0\in[-2,2]$, i.e., we modify the standard Lebesgue measure ($\lambda=\lambda_0$) such that its density is affine and more concentrated on $(0,1/2)$ if $b_0\in[-2,0)$ and more concentrated on $(1/2,1)$ if $b_0\in(0,2]$. It is straightforward to check that in \eqref{eq:biased mu}, $ b=b_0/12$. The case of Example \ref{ex:cent} is recovered when $b_0=0$.
\end{example}
It seems natural that the larger the positive (resp. negative) bias, the larger (resp. smaller) $\mu(\cC_I)$ will be. This happens because we are initially assigning less (resp. more) measure to the event $\{p<1/2\}$, i.e., to the event that the individual voter is more likely to choose wrongly. Therefore, Theorem \ref{th:unweig and indep}, as there is no bias ($b=0$), implies that for any measure $\mu=\prod_{n=1}^\infty \nu_n$ where the $\nu_n$ assign probability to both sides $\{p<1/2\}$ and $\{p>1/2\}$ ``fairly''\footnote{In the sense that $\int_{[0,1/2)}xd\nu_0(x)+\frac12\nu_0\left(\frac12\right)+\int_{(1/2,1]}xd\nu_0(x)=\frac12$. If all the mass were concentrated on $[0,1/2)$, $\nu_0[0,1/2)=1$, then previous sum would be $<1/2$ and the opposite for all the mass concentrated on $(1/2,1]$. }, then $\mu$ is going to assign measure zero to the CJP, i.e., the CJP will not hold almost surely. Hence, we get the same result as if $b<0$, see Example \ref{ex:bias measures th}. Sometimes, $b<0$ could be justified (e.g., \cite{Cap11}), but here we show that even if we assume $b=0$ because, following a Bayesian approach, we want to estimate the \textit{prior probability} (the probability before any evidence is collected) of the CJP, we will arrive at the same result: the CJP fails almost surely. That is, if we try to measure the applicability of the CJP according to a symmetrically balanced distribution (in particular, with no bias toward incompetence) without considering any evidence on voters competence, we arrive at the result that the CJT does not hold almost surely. Prior (or a priori in this case) probabilities are the baseline from which probabilities are updated when evidence is collected. So, in this setting, we would need strong evidence of voter competence to expect that the CJT can be applied. 

Nevertheless, the case $b<0$ has an important difference with respect to $b=0$. Now we can prove that, almost surely, the anti-CJP will hold, i.e.,
\begin{equation}\label{eq:antiCJT cond}
	\lim_{n\to\infty}\bP\left(\sum_{i=1}^n X_i<\frac{n}{2}\right)=1,
\end{equation}
\textit{the wrong option will be chosen almost surely}. Indeed, let $\Phi(x)=1-x$, $X'\coloneqq\Phi\circ X$ and $\nu_n'\coloneqq \Phi_*\nu_n$, the push-forward measure, as $\bP\left(X'=1\right)=:p'=\Phi(p)$ measures the probability of choosing the wrong option. Hence, for $\mu'=\prod_{n=1}^{\infty}\nu_n'$ we have
$$\int_{[0,1]}\lambda d\nu_0'(\lambda)=\int_{[0,1]}(1-\lambda) d\nu_0(\lambda)=\frac12-b=:1+b'>\frac12.$$ 
Thus, we can use Theorem \ref{th:unweig, indep and pos bias} to conclude the proof.

Therefore, the CJT is a \textit{double-edged sword}: it can either prove that majority rule is an almost perfect mechanism or an almost perfect disaster. This is partly the reason why we investigate here its applicability, to ensure that we are not in the case of a perfect disaster, but of a perfect mechanism to aggregate information.

As above, we should not confuse the Bayesian analysis in the two probability spaces. In the standard space, strategic voting and Bayesian--Nash equilibria were first analyzed by, among others, \cite{ASB96,McL98}. Our Bayesian approach is in the meta-probability level. We want to answer whether, given a dichotomous choice and a set of voters or jurors, we can invoke the CJT to ensure they will reach the correct option as the number of members increases. More precisely, we want to know its prior probability. These two problems are completely different.

The measures considered in Definition \ref{def:centered measure} constitute a quite general set of measures satisfying this symmetry condition of not favoring incompetence. Nevertheless, we can extend the results of Theorem \ref{th:unweig and indep} to a larger set of measures. This is treated in Theorem \ref{th:unweig and indep and biased} and Theorem \ref{th:extension}. As we have said, these more technical theorems extend Theorem \ref{th:unweig and indep} to some new measures, in particular, including the ones with $b\le 0$. 
\subsection{The case of $b>0$.}
Can $b>0$ be justified in some cases? As we commented in Section \ref{sec:intro}, we can achieve $p_i\ge 1/2$ for all $i\in\bN$ if the original voters with $p_i<1/2$ are assigned a weight of $-1$. In that case,
\begin{equation*}
	\bP\left((-1)X_i=1\right)=\bP\left(X_i=-1\right)=1-p_i>1/2,
\end{equation*}
where for simplicity we have assumed (see Section \ref{sec:intro}) that $X_i\in\{-1,1\}$. As we commented in the introduction, this does not seem easy to implement because voters can reject negative weights. Nevertheless, one could think that a rational voter will self-impose this if this voter knows that $p_i<1/2$. In other words, if the voter thinks the correct option is $A$ ($X_i=1$), then he/she votes $B$ ($X_i=-1$) and similarly for the opposite case. Now the probability is $p_i'=1-p_i>1/2$. But this strategy requires two steps:
\begin{itemize}
	\item knowing that $p_i<1/2$,
	\item be willing to reverse the outcome of one's vote. 
\end{itemize}
Considering real voters (not ideal ones), it is difficult to imagine the fulfillment of these steps. First, it is an empirical fact how well people calibrate their degree of knowledge with probability estimates. The standard finding of knowledge calibration experiments is overconfidence, people tend to overestimate their probability of being right, see Chapter 8 of \cite{StaRQ16} and references therein. And even if voters acknowledge that they are worse than a coin toss, it does not seem realistic to expect that, in general, they will reverse their outcome. For instance, they can rationalize their vote by introducing non-epistemic factors.

Be that as it may, our techniques can give us the result in this situation and this is the content of the following proposition:
\begin{theorem}\label{th:unweig, indep and pos bias}  Almost surely independent Condorcet Jury Theorem holds for a biased measure $\mu$ with $b>0$, that is
	\begin{equation}
		\mu(\mathcal{C}_I)=1.
	\end{equation}
\end{theorem}
\begin{proof}
	The proof of this theorem parallels the proof of Theorem \ref{th:unweig and indep} in Appendix \ref{ap:proof unweig}. The main change is that the denominator of $Q_n$ is:
	\begin{equation*}
		\sqrt{n}\times\frac{1}{n}\sum_{i=1}^n(p_i-1/2)
	\end{equation*}
	and the second factor tends to $b>0$ as $n\to\infty$ by the SLLN. Similarly, $m^1-m^2>0$ unless $\nu_0=\delta_1$, but in that case the theorem is trivial. Kakutani's lemma is applied in the same way.

\end{proof}

\section{The weighted CJT holds almost surely in almost every situation}\label{sec:weighted CJT}
But even in the case $b\le0$, not everything is lost. We can try to modify the aggregation procedures to achieve a competent mechanism. The natural idea is the consideration of weighted majority rules as described in \eqref{eq:weighted maj}.
By hypothesis, $\bP(X_i=1)=p_i$ and $\{X_i\}_{i=1}^\infty$ are independent. With weighted majority rule we have the following version of the CJT (sufficient conditions).
\begin{proposition}\label{prop:weighted CJT} If either
	\begin{equation}\label{eq:first cond weight}
		\frac{\sum_{i=1}^n w_i(p_i-q_i)}{\sqrt{\sum_{i=1}^n w_i^2p_iq_i}}\to \infty
	\end{equation}
	or, for any $n$ large enough,
	\begin{equation*}
		\sum_{i=1}^n w_i \delta_{p_i 1}>\sum_{i=1}^n w_i (1-\delta_{p_i 1})
	\end{equation*} 
	where $\delta_{ij}$ is the Kronecker delta, then the Condorcet Jury Property holds for the weights $w$. 
\end{proposition}
\begin{proof}
	As we have that
	\begin{equation*}
		\bE\left(X_i\right)=2p_i-1=p_i-q_i,~~\text{Var}\left(X_i\right)=4p_iq_i,
	\end{equation*}
	where $q_i\coloneqq 1-p_i$, then
	\begin{align*}
		\bP\left(X_n^w\le 0\right)&=\bP\left(X_n^w-\bE\left(X_n^w\right)\le -\bE\left(X_n^w\right)\right)\le \bP\left(|X_n^w-\bE\left(X_n^w\right)|\ge \bE\left(X_n^w\right)\right)\le\\
		&\le \frac{\text{Var}(X^w_n)}{ \bE\left(X_n^w\right)^2}=\frac{4\sum_{i=1}^n w_i^2p_iq_i}{\left(\sum_{i=1}^n w_i(p_i-q_i)\right)^2}\to 0
	\end{align*}
	as $n\to\infty$. We have used Chebyshev's inequality, which can be applied because $\bE\left(X_n^w\right)>0$ if $n$ is large enough by \eqref{eq:first cond weight}. For the second condition,
	\begin{equation*}
		X_n^w\ge \sum_{i=1}^n \left(w_i \delta_{p_i 1}-w_i(1-\delta_{p_i 1})\right)>0~~\text{a.s.}
	\end{equation*}
	by hypothesis if $n$ is large enough.
\end{proof}
\begin{remark}
	These conditions are the generalizations of \eqref{eq:first cond} and \eqref{eq:second cond}, see also \eqref{eq:ex1}, \eqref{eq:ex2} for the idea. If we had that, for instance, $w_i\ge 1$, it can be checked that the proof given in \cite{BP98} applies to our case and these conditions are necessary too. $w_i\ge 1$ corresponds to the case where no voter loses, formally, its weight on the election. 
\end{remark}

The next step would be to obtain, for some positive integer $k$ and constants $\alpha,\beta>0$,
\begin{equation}\label{eq:weight corr}
	w=\alpha+\beta p^k+\varepsilon,
\end{equation}
i.e., competence is positively correlated with the weight we assign but the association is not perfect, there is a stochastic error $\varepsilon$. More generally, \begin{equation*}
	w=\alpha+\sum_{i=1}^L\beta_i p^i+\varepsilon,
\end{equation*}
i.e., the weight is correlated with the probability of ``being right''(we assume the polynomial of $p$ is increasing), but there is a random error $\varepsilon$. This error can be interpreted as a measurement error, we cannot expect to obtain a perfect correlation. For simplicity we can assume \eqref{eq:weight corr} for some positive $k\in\mathbb{N}$. We also assume that $w\in [1,W]$ for some $W\ge1$. Thus, we choose $\alpha=1$ and $\beta=W-1$. That is, 
\begin{equation*}
	w=w_d(p)+\varepsilon
\end{equation*}
where $w_d$ would be the deterministic weight for a given probability $p$ going from $1$ to $W$ as a polynomial function. But there will be errors in the assignment of the weights and this is captured by $\varepsilon$.

In Theorem \ref{thm:measure weighted CJT} we show that if \eqref{eq:weight corr} is good enough, the CJT will hold almost surely for ``almost'' every measure $\mu$, even if they are strongly biased toward $p=0$, i.e., we are not only considering centered measures but the less favorable case of measures representing voters far from competence. In other words, we are not estimating the prior probability but the probability given almost any evidence on voters competence. This gives some evidence for trying to include epistemic weights in the decision procedure if we are interested in choosing the correct option. Our only requirement will be much weaker; basically $\nu_0\left((\frac12,1]\right)>0$. Otherwise ($p>\frac12$ does not happen almost surely), we cannot expect the CJP because in the best situation weights would reduce it to the case $\nu_0=\delta_{1/2}$ where we know that the CJP fails. As we see, this requirement is much weaker than imposing that the measure $\nu$ is centered, i.e., we allow the situation $m<1/2$ or $b<0$.

But even though the distribution might be biased toward the wrong option, we will prove that the CJT will hold almost surely if the weights are properly chosen. This is the content of the next theorem. We define $\cC^w_I$ as the set of sequences of probabilities $\{p_n\}_{n=1}^\infty$ such that for the weighted majority rule \eqref{eq:weighted maj} according to \eqref{eq:weight corr}, the CJP holds (note that the social choice function is not fixed as it depends on the weights). Also, $\tilde{\nu}_0$ is now a measure on $(p,\varepsilon)$ but $\varepsilon$ is not independent of $p$ (see \eqref{eq:def eps}), i.e., it is not a product measure. Similarly, $\mu$ will be absolutely continuous (see Section \ref{sec:notation}) w.r.t. $\mu_0=\prod_{n=1}^{\infty}\tilde{\nu}_0$. With this setting:
\begin{theorem}\label{thm:measure weighted CJT}
	Let $\{\varepsilon_n\}_{n=1}^\infty$ be a set of random variables distributed according to:
	\begin{equation}\label{eq:def eps}
		\varepsilon|_p\sim \cN_a^b(0,\sigma_W^2)
	\end{equation}
	where $\cN_a^b$ is the truncated Gaussian distribution restricted to the interval $(a,b)$ where $a\equiv a(p,W)\coloneqq-(W-1)p$, $b\equiv b(p,W)\coloneqq (W-1)(1-p)$. Let us assume that $\tilde{\nu}_0$ satisfies $\tilde{\nu}_0\left(\{p\in(\frac12,1]\}\right)>0$. Then, there is a $k$ such that if $(W-1)/\sigma_W, W$ are large enough, then
	$$
	\mu(\cC^w_I)=1,
	$$
	i.e., the CJP holds almost surely with this weighted majority rule.
\end{theorem}
The idea behind the theorem is clear: if we can find a procedure, with a suitable error, to assign weights according to competence, the (weighted) CJT will hold almost surely. Note that now we are considering the posterior probability too, as the measure is not centered any more. But, as we said, the main ingredient is the correct assignment of weights. In Appendix \ref{ap:psy and phil} we take this question seriously because there is little point in theorizing about something which cannot be practically implemented. We also consider how ``fair'' this situation would be.

\section{Concluding remarks}\label{sec:conclusion}
We have shown that the asymptotic CJP or the CJT for independent voters (which includes the MoA and the case studied by Condorcet) are, a priori, highly unlikely (see Theorem \ref{th:unweig and indep}, \ref{th:unweig and indep and biased}) unless we add some  good enough epistemic weights, i.e., weights correlated with epistemic rationality. That is, if we choose an arbitrary  sequence of voters, it will not satisfy the CJP almost surely. The bottom line is that applying the CJT (as it is common in some debates) might not be adequate if, using the Bayesian approach, there is no particular evidence of voter competence to update our priors (it might be the opposite case, \cite{Cap11}) nor some weights to correct the lack of competence. If ``good'' epistemic weights are added, its probability goes to one by Theorem \ref{thm:measure weighted CJT}. Note that in this latter case we are not estimating the prior probability but the probability given almost any evidence on voters competence (including the less favorable situations). These weights must be correlated (not necessarily a perfect correlation) with epistemic rationality and they guarantee a minimal weight of one to every voter. The CJT is an important and useful result to improve the decision-making process, but we have to ensure it holds when it is supposed to hold.

Obviously, our framework is a toy model of the real world, but a good point to start and, in fact, it is the same model that is usually used when the CJT is invoked. Some complications can be added and could be the topic of future research. For instance, in some processes we do not expect the options to remain unchanged if competent voters are more influential, but this is not directly captured in a dichotomous choice. An important limitation of the framework is the independence assumption. Votes can be correlated because of a deliberation process (``contagion'' in general), common sources of information or strategic voting, see \cite{Piv17} and references therein. Some works have treated the CJT for dependent voters, see for instance \cite{PZ12,Piv17}. In this case the known necessary and sufficient conditions involve the covariance between votes, say $\rho_{ij}$. Thus, the measure $\mu$ should include these parameters, but it cannot be a product measure as above. Indeed, if, for instance, $X_i\in\{0,1\}$, then
\begin{equation*}
p_{ij}\coloneqq \bE\left(X_i X_j\right)=\left|\bE\left(X_i X_j\right)\right|\le \sqrt{p_i p_j}
\end{equation*}
by H\"older's inequality. Also, as $p_{ij}=\bP\left(\{X_i=1\}\cap\{X_j=1\}\right)\le p_i,p_j$ and $\rho_{ij}=p_{ij}-p_ip_j$. So, $p_{ij}$ or $\rho_{ij}$ cannot be taken independently of $p_{i},p_j$ in $\mu$. A careful analysis would be needed in this situation because how to choose the measure is not trivial. Furthermore, sufficient and necessary conditions are less understood, so more analysis in that direction would be needed too. Anyway, this setting is somehow more restrictive as we not only need some competence condition, but additional requirements must be added, see \cite[Theorem 5.3]{Piv17}. For instance, the ``average correlation'' cannot growth too much: if votes are highly correlated there is little point in increasing the number of voters as they will vote in the same direction. So in this case, we would have to worry not only about competence but also about the correlation between votes.
\section{Acknowledgments}
The author would like to thank C. Romaniega for valuable comments which help to improve the manuscript.  The author is supported by the grant  MTM2016-76072-P (MICINN). This
work is supported in part by the ICMAT–Severo Ochoa grant CEX2019-000904-S. The author is
also a postgraduate fellow of the City Council of Madrid at the Residencia de Estudiantes
(2020--2022).
\newpage
\appendix
\section{Proof of Theorem \ref{th:unweig and indep} and Proposition \ref{prop:MoA cond}}\label{ap:proof unweig}
Let us assume first that $\mu_0=\prod_{n=1}^{\infty}\nu_0$. By \cite[Theorem 2]{BP98}, the CJP fails iff both
\begin{equation}\label{eq:first cond}
	\lim_{n\to\infty}\frac{\sum_{i=1}^n p_i-\frac{n}{2}}{\sqrt{\sum_{i=1}^n{p_iq_i}}}=\infty
\end{equation}
where $q_i\coloneqq 1-p_i$, and $\exists~n_0\in\bN$ such that
\begin{equation}\label{eq:second cond}
	|\{i:1\le i\le n \text{ and } p_i=1\}|>n/2\quad\foralll n>n_0
\end{equation}
do not hold. See \eqref{eq:ex1}, \eqref{eq:ex2} for the idea. For the first condition, we define:
\begin{equation*}
	Q_n\coloneqq\frac{\sum_{i=1}^n p_i-\frac{n}{2}}{\sqrt{\sum_{i=1}^n p_iq_i}}=\frac{\frac{1}{\sqrt{n}}\sum_{i=1}^n \left(p_i-\frac{1}{2}\right)}{\sqrt{\frac{1}{n}\sum_{i=1}^n{p_iq_i}}}
\end{equation*}
with $\sum_{i=1}^n{p_iq_i}=\sum_{i=1}^n{\left(p_i-p_i^2\right)}$. The \textit{key} here is to realize that under the measure $\mu_0$, $p_i$ are i.i.d. random variables in $([0,1]^\bN,\mu_0)$. Thus, we can apply the Strong Law of Large Numbers (SLLN), \cite[Theorem 10.13]{Fol99},
\begin{equation*}
	\frac{1}{n}\sum_{i=1}^n{\left(p_i-p_i^2\right)}\to b+\frac{1}{2}-m^2=m^1-m^2,
\end{equation*}
$\mu_0$-almost surely. Now, note that
\begin{equation}\label{eq:no conc}
	m^1-m^2=\int_{[0,1]} (x-x^2)~ d\nu_0 (x)=\int_{(0,1)} (x-x^2)~ d\nu_0 (x)>0
\end{equation}
as long as $\nu_0((0,1))>0$. This is going to be the case if $\epsilon_1<1/2$, as
\begin{equation}
	\frac12=m^1=\epsilon_1 +\int_{(0,1)} x~ d\nu_0 (x)
\end{equation} 
so $\int_{(0,1)} x~ d\nu_0 (x)=1/2-\epsilon_1>0$. For the numerator, we need the other classical asymptotic result, the Central Limit Theorem (CLT), to conclude 
\begin{equation*}
	\frac{1}{\sqrt{n}}\sum_{i=1}^n\left( p_i-\frac{1}{2}\right)\to \cN(0,\sigma^2)
\end{equation*}
in distribution as $n\to\infty$. Thus, by Slutsky Theorem \cite[Theorem 1.11]{Sha03}, 
\begin{equation*}
	Q_n\to \cN(0,\sigma'^2)\quad \text{where }~ \sigma'\coloneqq \frac{\sigma}{\sqrt{m^1-m^2}}
\end{equation*}
in distribution as $n\to\infty$. Let $Q\sim \cN(0,\sigma'^2)$. So we can conclude that,
\begin{equation}\label{eq:prob Qn inf}
	\mu_0\left(\lim_{n\to\infty}Q_n=\infty\right)=0.
\end{equation}
Indeed, let us define the events
$$
\cA\coloneqq \{(p_i)_{i=1}^\infty\in[0,1]^\bN~/~Q_n((p_i)_{i=1}^\infty)\to\infty\},\quad 
\cA_{n,\varepsilon}\coloneqq \{(p_i)_{i=1}^\infty\in[0,1]^\bN~/~Q_n((p_i)_{i=1}^\infty)>M_\varepsilon\}
$$
where $M_\varepsilon$ satisfies $\mu_0\left(Q\le M_\varepsilon\right)=1-\varepsilon$. By the definition of limit, for every $\varepsilon>0$,
\begin{equation*}
	\cA\subset \bigcup_{N\in\bN}\bigcap_{n\ge N} \cA_{n,\varepsilon}.
\end{equation*}
By the continuity of measures, \cite[Theorem 1.8.c)]{Fol99}
\begin{equation*}
	\mu_0\left(\bigcup_{N\in\bN}\bigcap_{n\ge N} \cA_{n,\varepsilon}\right)=\lim_{N\to\infty}\mu_0\left(\bigcap_{n\ge N}\cA_{n,\varepsilon}\right)\le\lim_{N\to\infty}\mu_0\left(\cA_{N,\varepsilon}\right)= \varepsilon,
\end{equation*}
where the last equality follows from the convergence in distribution. Hence, by the monotonicity of measures,
\begin{equation*}
	\mu_0\left(\cA\right)\le \varepsilon ~~\foralll\varepsilon>0,
\end{equation*}
concluding the proof of \eqref{eq:prob Qn inf}. Thus, the first condition \eqref{eq:first cond} will not hold almost surely. Let us see the second one \eqref{eq:second cond}. For that purpose, let us define for $p\in[0,1]$:
\begin{equation*}
	\tilde{p}:=\begin{cases*}
		1 \text{ if } p=1\\
		0 \text{ if } p\in[0,1)
	\end{cases*}
\end{equation*}
Thus, if $\mu_0(p_i\in A)=\nu_0(A)$ for $A$ a Borel set, $\tilde{p}_i\sim Bernoulli(\epsilon_1)$ where we defined $\epsilon_1$ as $\nu_0(\{1\})$. Let us also define
\begin{equation}\label{eq:def Sn}
	S_n\coloneqq|\{i:1\le i\le n \text{ and } p_i=1\}|=\sum_{i=1}^n \tilde{p}_i
\end{equation} 
and 
\begin{align}\label{eq:B def}
	\mathcal{B}&\coloneqq\{(p_i)_{i=1}^\infty\in[0,1]^\bN~/~\existss n_0 : S_{n_0+2k}>n_0/2+k~\foralll k\ge 0\}\text{,  }\nonumber\\ \mathcal{B}_n&\coloneqq\{(p_i)_{i=1}^\infty\in[0,1]^\bN~/~ S_n>n/2\}.
\end{align} 
By the SLLN,
\begin{equation*}
	\mu_0\left(\lim_{n\to\infty}\frac{S_n}{n} =\epsilon_1<\frac12\right)=1.
\end{equation*}
But if $\lim_{n\to\infty}S_n/n=\epsilon_1<\frac12$, then
\begin{equation*}
	\frac{S_n}{n}<\frac12
\end{equation*}
if $n$ is large enough. Thus, $\mu_0(\cB)=0.$
Summing up, 
\begin{equation}\label{eq:mu0 cCI}
	\mu_0(\cC_I)=\mu_0(\cA\cup\cB)=0
\end{equation}
concluding the proof of the theorem for $\mu_0=\prod_{n=1}^{\infty}\nu_0$ if $\epsilon_1<\frac12$. \\

We consider now the case of $\epsilon_1=1/2$ and therefore $\nu_0=\frac12 (\delta_0+\delta_1)$. First note that condition \eqref{eq:first cond} does not hold because here either $p_i=0$ or $q_i=0$ almost surely, so $p_i=0$ or $q_i=0$ $\foralll i\in\bN$ almost surely, i.e., $Q_n$ is not well-defined almost surely. But in this deterministic (only $p=0$ or 1) situation, it is clear that CJP holds iff \eqref{eq:second cond} holds. So let us show that the former fails.
Define $\bar{p}\coloneqq1$ if $p=1$ and $\bar{p}\coloneqq-1$ otherwise. If $\bar{S}_{n_0}^n\coloneqq\sum_{i=n_0}^n \bar{p}_i$, then, $\bar{S}_{n_0}^n$ is a symmetric random walk in $n$ starting at zero. If we denote, for $k\in\mathbb{Z}$,
\begin{equation*}
	r_k\coloneqq \mu_0\left(\{(p_i)_{i=1}^\infty\text{ such that }\exists~n~/~ \mathcal{S}_n=k\}\right)
\end{equation*}
where $\mathcal{S}_n$ is a symmetric random walk starting at zero. It is standard that $r_0=1$, in fact, the probability of returning infinitely often to 0 equals 1. Then, we have the following difference equation:
\begin{equation*}
	r_k=\frac{1}{2}r_{k-1}+\frac{1}{2}r_{k+1},
\end{equation*}
i.e., if the first move is up, there are $k-1$ up movements left and, similarly, if the first movement is down, we would need $k+1$ movements up. The solution is $r_k=c\in[0,1]$ because $0\le r_k\le 1$. As $r_0=1$, we conclude $r_k=1$ $\foralll k\in\mathbb{Z}$.

Now, fix a $n_0\in\bN$. Then, if $\bar{S}_{n}\coloneqq \bar{S}_{1}^{n}$ and \eqref{eq:second cond} holds for that $n_0$, then
\begin{equation*}
	\bar{S}_{n_0}=i
\end{equation*}
for $i=1,...,n_0$ such that $\bar{S}_n>0$ for $n\ge n_0$ odd. But, by the discussion above, the probability that $\bar{S}_{n_0}^n<-i<0$ is one, so the probability that $S_n>n/2$ given a fixed $S_{n_0}\ge\frac{n_0+1}{2}$ is zero. Therefore, the negation of \eqref{eq:second cond} almost surely as $\mu_0\left(\cB\right)\le\sum_{n_0}\mu_0(\cap_{k\ge 0} \mathcal{B}_{n_0+2k})=0,$ where
$$
\mathcal{B}=\bigcup_{n_0\in\mathbb{O}  }\bigcap_{k\ge 0} \mathcal{B}_{n_0+2k}.
$$
We need the following technical lemma to conclude the proof for a general centered measure. 
\begin{lemma}\label{lem:Kakutani} $\mu=\prod_{n=1}^{\infty}\nu_n\ll\mu_0=\prod_{n=1}^{\infty}\nu_0$ provided \eqref{eq:dist cond} holds.	
\end{lemma}
With this lemma the proof of Theorem \ref{th:unweig and indep} is concluded as $\mu(\cC_I)=0$ by \eqref{eq:mu0 cCI}.
Before proving the lemma we present briefly some tools that will be needed.
\subsection{Distances and divergences}\label{sec:dist and div}
First, some definitions. Consider a family $M$ of probability distributions or measures.
\begin{definition}[Divergence] Let $M$ be as above and suppose that we are given a (smooth) function $d( \cdot|| \cdot) : M \times M \to \mathbb{R}$ satisfying the following properties
	$\forall$ $p,q \in M$:
	\begin{itemize}
		\item[i)] $d(p||q)\ge 0$, 
		\item[ii)] and $d(p||q)=0$ iff $p =q$.
	\end{itemize}
	Then, $d$ is said to be a {divergence.}
	\label{def:Div}
\end{definition}
\begin{definition}[Distance] Let $M$ be as above and suppose that we are given a function $d( \cdot, \cdot) : M \times M \to \mathbb{R}$ satisfying the following properties
	$\forall$ $p,q,r \in M$:
	\begin{itemize}
		\item[i)](Positive definiteness) $d(p,q)\ge 0$, and $d(p,q)=0$ iff $p =q$,
		\item[ii)](Symmetry) $d(p,q)=d(q,p)$,
		\item[iii)](Triangle inequality) $d(p,r)\le d(p,q)+d(q,r)$.
	\end{itemize}
	Then, $d$ is said to be a {distance.}
	\label{def:distance}
\end{definition}
Sometimes we will also use the notation $d( \cdot, \cdot)$ for divergences too. As a divergence do not necessarily satisfies ii) and iii), it is usually called a pseudodistance. Let us explore some examples. Given two measures $\mu,\mu'$ defined on the measurable space ($X,\Sigma$), the \textit{total variation distance} will be denoted by:
\begin{equation*}
	\norm{\mu-\mu'}\coloneqq 2 \ 
	\sup _ {B \in {\Sigma} } \ 
	| \mu(B) -
	\mu'( B) |.
\end{equation*}
That is, the total variation distance is twice the ``maximum'' difference between the measure of the same set for $\mu$ and $\mu'$. That is, we have the useful bound
\begin{equation}
	| \mu(B) -	\mu'( B) |\le\frac12 \norm{\mu-\mu'}\quad\foralll B \in {\Sigma},
\end{equation}
so the smaller $\norm{\mu-\mu'}$, the smaller the discrepancy between $\mu(B)$ and $\mu'( B)$ for every measurable set. In fact, it can be shown this is a norm in the space of Radon signed measures, \cite[Proposition 7.16]{Fol99}.
It is a well-known identity that:
\begin{equation}\label{eq:L1 identity}
	2\sup _ {B \in {\Sigma} } \left|{\mu(A) - \mu'(A)}\right| = \norm{ \rho - \rho' }_{L_1(\tau)}\coloneqq \int_{X} |\rho(x)-\rho'(x)|d\tau(x),
\end{equation}
where $\rho\coloneqq d\mu/d\tau$ and $\rho'\coloneqq d\mu'/d\tau$ for some $\tau\gg \mu,\mu'$. For instance, $\tau\coloneqq\frac12\left(\mu+\mu'\right) $.
Also, the relative entropy or Kullback–Leibler divergence is defined as:
\begin{equation*}
	d_{KL}(\mu\|\mu') \coloneqq \int_{X} \log \frac{\rho(x)}{\rho'(x)}\,\rho(x)d\tau(x).
\end{equation*}
See \cite{RRT19} for more details and for a theoretical framework relating divergences and entropies and for more (geometrical) properties of divergences.
\begin{proof}[Proof of the technical Lemma \textnormal{\ref{lem:Kakutani}}]
	Let
	\begin{equation}\label{eq:def dH}
		d_H(\nu_n,\nu_0)\coloneqq\left( 2 \left(1 - H ( \nu_n, \nu_0)\right) \right)  ^ {1/2}=\left( \int_X \left ( \sqrt {
			\frac{d \nu_0 }{d \tau }
		} - \sqrt {
			\frac{d \nu_n }{d \tau }
		} \right )  ^ {2}  d \tau \right)  ^ {1/2}
	\end{equation}
	where
	$$ 
	H (\nu_0,\nu_n)  \coloneqq \ 
	\int\limits _ {  X }
	\sqrt {
		\frac{d{\nu_n} }{d \tau }
	}
	\sqrt {
		\frac{d {\nu_0 } }{d \tau }
	}  d \tau
	$$
	is the Hellinger integral with $\tau$ is a measure such that $\nu_0,\nu_n$ are absolutely continuous and $X=[0,1]$ here. By Kakutani Dichotomy Theorem, \cite[Proposition 2.21]{dPZ14}, if
	\begin{equation}\label{eq:pos prod}
		\prod_{n=1}^\infty H(\nu_0,\nu_n)>0,
	\end{equation}
	then $\mu\ll\mu_0$. 
	To prove \eqref{eq:pos prod} we need to know the following fact: for $0\le a_n<1$, $\prod _{{n=1}}^{{\infty }}(1-a_{n})$ converges to a positive number iff $-\sum_n \log(1-a_n)$ converges iff $\sum_n a_n$ converges, by the limit comparison test. Here $1-a_n=H(\nu_n,\nu_0)$. Indeed, $H(\nu_n,\nu_0)\le1$ by Cauchy-Schwarz's inequality and $H(\nu_n,\nu_0)>0$ as $\nu_n\ll\nu_0$ by hypothesis (so take $\tau=\nu_0$). Thus, it is enough if we prove that
	\begin{equation}\label{eq:dom cond}
		1-H(\nu_n,\nu_0)\lesssim \norm{\nu_n-\nu_0}, d_{KL}(\nu_n-\nu_0)
	\end{equation} 
	which entails that $\sum_n\left(1-H(\nu_n,\nu_0)\right)$ converges by \eqref{eq:dist cond}. First, using \eqref{eq:L1 identity}, \eqref{eq:def dH} and
	$$
	\left ( \sqrt {
		\frac{d \nu }{d \tau }
	} - \sqrt {
		\frac{d \nu' }{d \tau }
	} \right )  ^ {2}(x)\le\left ( \sqrt {
		\frac{d \nu }{d \tau }
	} - \sqrt {
		\frac{d \nu' }{d \tau }
	} \right )(x)  \left ( \sqrt {
		\frac{d \nu }{d \tau }
	} + \sqrt {
		\frac{d \nu' }{d \tau }
	} \right )(x) =\left | 
	\frac{d \nu }{d \tau }
	- 
	\frac{d \nu' }{d \tau }
	\right |(x).
	$$
	assuming w.l.o.g. that $\left(\sqrt {\frac{d \nu }{d \tau }
	} - \sqrt {	\frac{d \nu' }{d \tau }}\right)(x)\ge 0$. For the second, 
	\begin{align*}
		d_{KL}(\nu\|\nu') &= \int \log \frac{\rho(x)}{\rho'(x)}\,\rho(x)dx= 2\int \log \frac{\sqrt{\rho(x)}}{\sqrt{\rho'(x)}}\,\rho(x)dx\\
		&= 2\int -\log \frac{\sqrt{\rho'(x)}}{\sqrt{\rho(x)}}\,\rho(x)dx\ge 2\int \left(1-\frac{\sqrt{\rho'(x)}}{\sqrt{\rho(x)}}\right)\,\rho(x)dx\\
		&= \int \left(1+1-2\sqrt{\rho(x)}\sqrt{\rho'(x)}\right)\,dx= \int \left(\sqrt{\rho(x)}-\sqrt{\rho'(x)}\right)^2\,dx= 1-H(\nu,\nu'),
	\end{align*}
	where we have used that $-\log(x)\ge 1-x$ and defined $\rho\coloneqq \frac{d\nu}{d\tau}$ and $\rho'\coloneqq \frac{d\nu'}{d\tau}$.
\end{proof}
As we see from \eqref{eq:dom cond}, it is enough for our purposes if the distance satisfies
\begin{equation}\label{eq:d condition}
	1-H(\nu_n,\nu_0)\lesssim  d(\nu_n-\nu_0).
\end{equation}
So this is the condition which defines $\cD$. More specifically, 
\begin{definition}\label{rem:def cD} Let $d$ a distance or divergence as above, then $d\in\cD$ iff $d$ satisfies \eqref{eq:d condition}.
\end{definition}
\begin{example}
	In the last part of the proof we showed that $d_{KL},\norm{\cdot}\in\cD$ so defined. But the set is larger than that, for instance,  Bhattacharyya distance is defined as:
	\begin{equation*}
		d_B(\nu,\nu')\coloneqq -\log H(\nu,\nu').
	\end{equation*}
	Then, $d_B\in\cD$ because $-\log(x)\ge 1-x$
\end{example}
We finish this section with the proof of Proposition \ref{prop:MoA cond}. We define:
\begin{equation*}
	\hat{p}^{\varepsilon_0}:=\begin{cases*}
		1 \text{ if } p\in[1-\varepsilon_0,1]\\
		0 \text{ if } p\in[0,1-\varepsilon_0)
	\end{cases*}
 \end{equation*}
with $\hat{p}^0=\tilde{p}$. Thus, we have that
\begin{equation*}
	\lim_{n\to\infty}\frac{1}{n}\sum_{i=1}^n \hat{p}^{\varepsilon_0}_i=\epsilon_{1-\varepsilon_0,1} ~~ \text{a.s.}
\end{equation*}
hence, by Egorov's Theorem the convergence is almost uniform. If $0<\varepsilon<\epsilon_{1-\varepsilon_0,1}$, then $\varepsilon'\coloneqq\epsilon_{1-\varepsilon_0,1}-\varepsilon>0$ so by the definition of limit for $n>N$ large enough
\begin{equation*}
	\big\lvert	\frac{1}{n}\sum_{i=1}^n \tilde{p}_i -\epsilon_{1-\varepsilon_0,1}\big\rvert<\varepsilon'\Rightarrow \frac{1}{n}\sum_{i=1}^n \tilde{p}_i> \varepsilon.
\end{equation*} 
The result follows from the fact that almost uniform convergence implies that this happens (with $N$ uniform) for a set of measure no less than $1-\delta$.

\section{Extending Theorem \ref{th:unweig and indep}, proofs and an example}\label{ap:extension}
We present a theorem which includes some cases not considered in Theorem \ref{th:unweig and indep}. There is some overlapping with Theorem \ref{th:unweig and indep} but we opted to give a self-contained and easier proof of that theorem in Appendix \ref{ap:proof unweig}. This makes the exposition clearer for some readers as in the next proof we will use more technical tools.
\begin{theorem}\label{th:unweig and indep and biased} If $\mu=\prod_{i=1}^{\infty}\nu_i$ is a measure such that:
	\begin{align}
		&\limsup_{n\to\infty}\frac{1}{\sqrt{n}}\sum_{i=1}^{n}\left( m_i-\frac{1}{2}\right)<\infty,\label{eq:gen cent cond}\\
		&\liminf_{n\to\infty}\frac{1}{{n}}\sum_{i=1}^{n}\left( m_i-m_i^2\right)>0,\label{eq:gen no conc}	\\
		&\limsup_{n\to\infty
		}{\frac{1}{n}\sum_{i=1}^n \epsilon_{1i}}<\frac12, \label{eq:gen eps1}
	\end{align}
	where $\epsilon_{1i}\coloneqq\nu_i(\{1\})$ and $\sigma_{T,n}\coloneqq\left(\sum_{i=1}^n \bE\left((p_i-m_i)^2\right)\right)^{\frac12}$ goes to infinity, then CJP does not hold $\mu$-almost surely, i.e., $\mu(\cC_I)=0$. 
\end{theorem}
\begin{example}\label{ex:bias measures th}
	The biased measures of Definition \ref{def:biased measure} are included in this theorem. Indeed, for $\mu_0$ with $b<0$, $m^1<\frac12$ so \eqref{eq:gen cent cond} holds because $m_i=m^1~\forall i\in\bN$. Condition \eqref{eq:gen no conc} holds if $\nu_0\neq \alpha\delta_0+(1-\alpha)\delta_1$ for\footnote{We arrive to the same conclusion in this case of $\nu_0= \alpha\delta_0+(1-\alpha)\delta_1$, but we should argue as in \eqref{eq:def Sn} and below.} some $1/2<\alpha\le 1$. Indeed,
	\begin{equation}
		m^1-m^2=\int_{[0,1]} (x-x^2)~ d\nu_0 (x)=\int_{(0,1)} (x-x^2)~ d\nu_0 (x)>0
	\end{equation}
	as long as $\nu_0((0,1))>0$. Now, as $m^1<1/2$, it must be that $\nu_0\left(\{1\}\right)=\epsilon_1< \frac12$ so \eqref{eq:gen eps1} holds trivially.

\end{example}
We can also extend this theorem for the case 
\begin{equation*}
	\limsup_{n\to\infty
	}{\frac{1}{n}\sum_{i=1}^n \epsilon_{1i}}=\frac12.
\end{equation*}
First, some definitions, recall \eqref{eq:def Sn},
\begin{align*}
	&\mathcal{B}_{n_0,n}\coloneqq \bigcap_{k=0 }^{(n-n_0)/{2}}\cB_{n_0+2k}= \{(p_i)_{i=1}^\infty\in[0,1]^\bN~/~ S_k>k/2~~\foralll k\in\{n_0,n_0+2,\ldots,n\}\},\\ &\mathcal{B}_{n_0,n}^b\coloneqq \{(p_i)_{i=1}^\infty\in[0,1]^\bN~/~ S_k>k/2~~\foralll k\in\{n_0,n_0+2,\ldots,n\}\text{ and }S_{n}=(n+1)/{2}\}.
\end{align*}
The first set is given by the sequences such that the sum satisfies $S_k>k/2$ as in condition \eqref{eq:second cond} for odd numbers between $n_0$ and $n$ and the second is a subset such that the last sum is in the border case $S_n=(n+1)/2$.
\begin{theorem}\label{th:extension}
	Assume that \eqref{eq:gen cent cond} and \eqref{eq:gen no conc} hold. If the $\liminf$ in \eqref{eq:gen eps1} is $1/2$ substitute \eqref{eq:gen eps1} for either there is a $n\ge n_0$ such that $\mu\left(\mathcal{B}_{n_0,n}\backslash\mathcal{B}_{n_0,n}^b\right)=0$ and $\mu(\{S^{n+2}_{n+1}=0\})$ or
	\begin{equation}\label{eq:border growth cond}
		\sum_{k=0}^\infty\mu(\mathcal{B}_{n_0,n_0+2k}^b~|~ \mathcal{B}_{n_0,n_0+2k}) (1-\epsilon_{1(n_0+2k+1)})(1-\epsilon_{1(n_0+2k+2)})=\infty~~\forall n_0\in\mathbb{O}.
	\end{equation}
	If so, we arrive at the same conclusion, the CJP does not hold $\mu$-almost surely. 
\end{theorem}
\begin{remark}
	Some comments on the new hypothesis are in order. First, condition \eqref{eq:gen cent cond} is a generalization of the centered condition \eqref{eq:centered}. Second, condition \eqref{eq:gen no conc} is a generalization of $m^1>m^2$ that we saw in \eqref{eq:no conc}. Third, \eqref{eq:gen eps1} is a generalization of $\epsilon<1/2$ in the previous theorem. Condition \eqref{eq:border growth cond} can be used to treat the case of purely atomic measures like $\nu_0=\frac12 (\delta_0+\delta_1)$ where there is no uncertainty as either $p=0$ or $p=1$, see Appendix \ref{ap:extension}.
\end{remark}

The proof of Theorem \ref{th:unweig and indep and biased} is going to be similar (except Kakutani's Theorem) to the proof of Theorem \ref{th:unweig and indep}, but more technical. Some steps which are already there will be omitted here (but they will be properly referenced). As we did there, we define
\begin{equation*}
	Q_n\coloneqq\frac{\sum_{i=1}^n p_i-\frac{n}{2}}{\sqrt{\sum_{i=1}^n p_iq_i}}=\frac{\frac{1}{\sqrt{n}}\sum_{i=1}^n\left( p_i-\frac{1}{2}\right)}{\sqrt{\frac{1}{n}\sum_{i=1}^n{p_iq_i}}}.
\end{equation*}
As we said above, this quotient appears in the necessary and sufficient conditions of the CJT. 
$Q_n$ can be rewritten as:
\begin{equation*}
	Q_n=\frac{\frac{1}{\sqrt{n}}\sum_{i=1}^n\left( m_i-\frac{1}{2}\right)+\frac{\sigma_{T,n}}{\sqrt{n}}\frac{1}{\sigma_{T,n}}\sum_{i=1}^n\left( p_i-m_i\right)}{\sqrt{\frac{1}{n}\sum_{i=1}^n{\left(p_i-m_i-(p_i^2-m_i^2)+m_i-m_i^2\right)}}}.
\end{equation*}
where $\sigma_{T,n}\coloneqq\left(\sum_{i=1}^n \bE\left((p_i-m_i)^2\right)\right)^{\frac12}$. Now, by Kolmogorov's version of the SLLN (which we can apply because of \eqref{eq:moments}) we have almost surely:
\begin{equation*}
	\lim_{n\to\infty}\frac{1}{n}\sum_{i=1}^n{(p_i-m_i)}=0,\quad 	\lim_{n\to\infty}\frac{1}{n}\sum_{i=1}^n{(p_i^2-m_i^2)}=0.
\end{equation*}
Now, by the hypotheses of the theorem we can take a subsequence such that
\begin{align*}
	&\lim_{k\to\infty}\frac{1}{\sqrt{n_k}}\left(\sum_{i=1}^{n_k} m_i-\frac{1}{2}\right)=C<\infty,\\
	&\lim_{k\to\infty}\frac{1}{{n_k}}\left(\sum_{i=1}^{n_k} m_i-m_i^2\right)=c>0.	
\end{align*}
Finally, let us note that by \eqref{eq:moments},
\begin{equation*}
	\frac{\sigma_{T,n}}{\sqrt{n}}\le 1.
\end{equation*}
Similarly,
\begin{equation*}
	\frac{\sum_{i=1}^n\bE\left(p_i-m_i\right)^3}{\sigma_{T,n}^3}\le \frac{\sigma_{T,n}^2}{\sigma_{T,n}^3}\to 0
\end{equation*}
as $\sigma_{T,n}\to\infty$. Thus, Lyapunov's condition holds, so we can apply Lindeberg’s CLT. Hence, taking a subsequence (that we relabeled again) and using Slutsky Theorem as before,
\begin{equation*}
	Q_{n_k}\to \cN(\mu_Q,\sigma_Q^2)
\end{equation*}
in distribution as $k\to\infty$ for some $\mu_Q,\sigma_Q\in\bR$. Then we can apply \eqref{eq:prob Qn inf} to conclude that the first condition of \cite[Theorem 2]{BP98} does not hold almost surely.
Let us turn now to the second condition. For that purpose, as we did above, we define:
\begin{equation*}
	\tilde{p}:=\begin{cases*}
		1 \text{ if } p=1\\
		0 \text{ if } p\in[0,1)
	\end{cases*}
\end{equation*}
We also define the sums:
\begin{equation*}
	S_{n_0}^n\coloneqq|\{i:1\le n_0\le n \text{ and } p_i=1\}|=\sum_{i=n_0}^n \tilde{p}_i.
\end{equation*} 
First, if
\begin{equation}\label{eq:liminf eps}
	\limsup_{n\to\infty
	}{\frac{1}{n}\sum_{i=1}^n \epsilon_{1i}}<\frac12,
\end{equation}
by Kolmogorov's SLLN then,
\begin{equation*}
\lim_{n\to\infty}\left(\frac{1}{n}\sum_{i=1}^n \tilde{p}_i-\frac{1}{n}\sum_{i=1}^n \epsilon_{1i}\right)=	\limsup_{n\to\infty}\left(\frac{1}{n}\sum_{i=1}^n \tilde{p}_i-\frac{1}{n}\sum_{i=1}^n \epsilon_{1i}\right)=0~~\text{a.s}.
\end{equation*}
Using that $\limsup_{n\to\infty}(x_n)+\liminf_{n\to\infty}(y_n)\le\limsup_{n\to\infty}(x_n+y_n)$, then
\begin{equation*}
	\limsup_{n\to\infty}\frac{1}{n}\sum_{n=1}^\infty \tilde{p}_n= \limsup_{n\to\infty
	}{\frac{1}{n}\sum_{i=1}^n  \epsilon_{1i}}<\frac12
\end{equation*}
almost surely. Thus, for large enough $n_0$,
\begin{equation*}
	\sup_{n\ge n_0}	\frac{S_n}{n}<\frac{1}{2}-\delta
\end{equation*}
for some $\delta>0$, therefore violating the second condition for the CJT. This finishes the proof of Theorem \ref{th:unweig and indep and biased}.

Let us prove Theorem \ref{th:extension}. Recall that we defined the sets, for $n_0\le n$ odd numbers:
\begin{align*}
	&\mathcal{B}_{n_0,n}\coloneqq \bigcap_{k=0 }^{(n-n_0)/{2}}\cB_{n_0+2k}= \{(p_i)_{i=1}^\infty\in[0,1]^\bN~/~ S_k>k/2~~\foralll k\in\{n_0,n_0+2,\ldots,n\}\},\\ &\mathcal{B}_{n_0,n}^b\coloneqq \{(p_i)_{i=1}^\infty\in[0,1]^\bN~/~ S_k>k/2~~\foralll k\in\{n_0,n_0+2,\ldots,n\}\text{ and }S_{n}=(n+1)/{2}\}.
\end{align*}
If the $\liminf$ in \eqref{eq:liminf eps} is 1/2, note that
\begin{equation*}
	\mathcal{B}_{n_0,n+2}=\{(p_i)_{i=1}^\infty\in\mathcal{B}_{n_0,n}^b~/~S_{n+1}^{n+2}\ge 1\}\bigsqcup\left(\mathcal{B}_{n_0,n}\backslash\mathcal{B}_{n_0,n}^b\right),
\end{equation*}
where $\sqcup$ denotes a disjoint union. The idea is that a sequence is in $\mathcal{B}_{n_0,n+2}\subset \mathcal{B}_{n_0,n}$ because it satisfies either $S_{n_0}^{n}=(n+1)/2$ (so we need that the next two summands are at least 1) or  $S_{n_0}^{n}>(n+1)/2$ so the sequence is in $\mathcal{B}_{n_0,n+2}$, independently of the next two summands. Thus, applying the (product) measure we will obtain the following recurrence relation:
\begin{align*}
	\mu(\mathcal{B}_{n_0,n+2})&=\mu(\mathcal{B}_{n_0,n})\left(\mu(\mathcal{B}_{n_0,n}^b~|~ \mathcal{B}_{n_0,n})\mu\left(\{S_{n+1}^{n+2}\ge1\}\right)+	1-\mu(\mathcal{B}_{n_0,n}^b~|~ \mathcal{B}_{n_0,n})\right)\\
	&=\mu(\mathcal{B}_{n_0,n})\left(1-\alpha_{n_0,n}\beta_n\right),
\end{align*}
where
\begin{equation*}
	\mu(\mathcal{B}_{n_0,n}^b~|~ \mathcal{B}_{n_0,n})\coloneqq \frac{\mu\left(\mathcal{B}_{n_0,n}^b\cap \mathcal{B}_{n_0,n}\right)}{\mu\left(\mathcal{B}_{n_0,n}\right)}=:\alpha_{n_0,n}~\text{ and }~\beta_n\coloneqq \mu\left(\{S_{n+1}^{n+2}=0\}\right).
\end{equation*}
Note that $\beta_n=(1-\epsilon_{1(n+1)})(1-\epsilon_{1(n+2)})$. Thus, for $n>n_0$ both odd:
\begin{equation}\label{eq:border cond orig}
	\mu(\mathcal{B}_{n_0,n+2})=\mu(\mathcal{B}_{n_0,n_0})\prod_{k=0}^{\frac{n-n_0}{2}}(1-\alpha_{n_0,n_0+2k}\beta_{n_0+2k}).
\end{equation}
Now we take the limit $n\to\infty$. If there is some $k$ such that $(1-\alpha_{n_0,n_0+2k}\beta_{n_0+2k})=0$, the product is zero and this corresponds to the first condition of Theorem \ref{th:extension}. Otherwise, as we saw in the proof of Lemma \ref{lem:Kakutani}, this infinite product will be zero iff 
\begin{equation*}
	\sum_{k=0}^{\infty}\alpha_{n_0,n_0+2k}\beta_{n_0+2k}=\infty.
\end{equation*}
Then, similarly as we did above,
\begin{equation*}
	\mu\left(\cB_{n_0,\infty}\coloneqq\bigcap_{k=0}^\infty\cB_{n_0,n_0+2k}\right)=\lim_{k\to\infty}\mu\left(\cB_{n_0,n_0+2k}\right)=0.
\end{equation*}
Therefore, $\mu\left(\cB\right)\le\sum_{n_0}\mu(\cB_{n_0,\infty})=0.$
This finishes the proof.
\subsection{Example of application to the case $\nu_0=\frac12\left(\delta_0+\delta_1\right)$: some combinatorics}\label{sec:example comb}

Let us consider how to apply \eqref{eq:border growth cond}. For simplicity, assume $n_0=1$. In order to understand the set $\mathcal{B}_{1,n}^b$, we need to know how many points will satisfy $S_{1+2k}\ge k+1$ for $k=0,\ldots,(n-1)/2$ and $S_{n}=(n+1)/2$. For $k=0$, the only possibility is $S_1=1$. Thus, we need to see the number of ways in which $S_2^{n}=\frac{n-1}{2}$ such that $S_2^{1+2k}\ge k$ for $k=0,\ldots,(n-1)/2-1$. We can see this graphically if we consider a grid where $\tilde{p}_i=1$ is translated into moving up and $\tilde{p}_i=0$ is translated into moving to the right. The conditions above are equivalent to the condition that for every point $(x,y)$ of the path such that $x+y=2k$, then $x\le y$ and  the end point is $(m,m)$ where $n=2m+1$. This is illustrated in Figure \ref{fig:Fig1}. Note that the blue path satisfies these conditions while the red one does not because of the point $(3,1)$ in black.
\begin{figure}[h]
	\centering
	\includegraphics[width=0.5\linewidth]{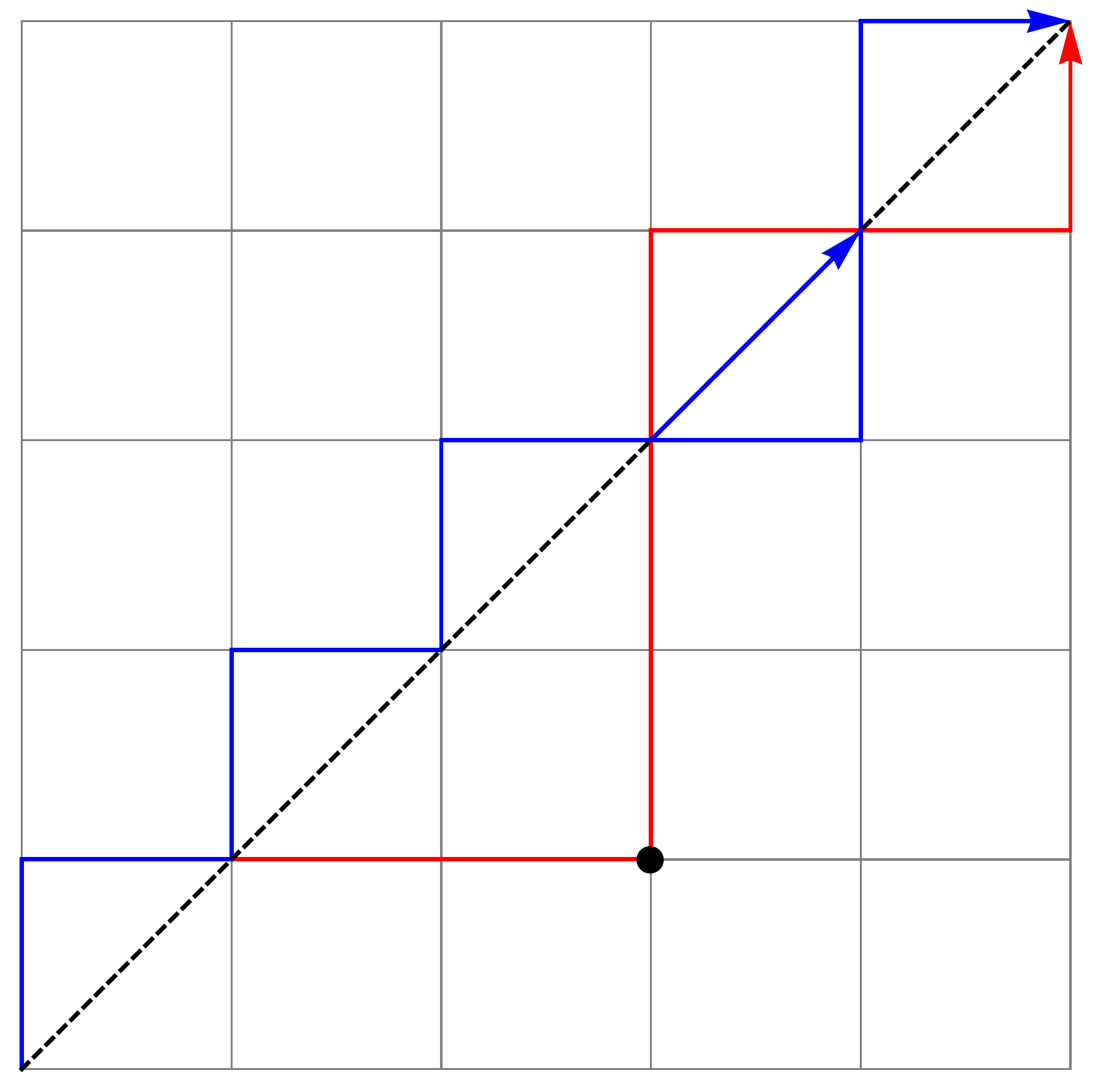}
	\caption{Illustration for $n=11$.}
	\label{fig:Fig1}
\end{figure}
Our conditions are not the same as lying above the diagonal (dashed line), as the blue path is below it at $(4,3)$. Nevertheless, if we allow the path to move $(1,1)$ at points on the diagonal, like the blue arrow in Figure \ref{fig:Fig1} shows, we can consider that the allowed paths are always above the diagonal. So the problem reduces to counting the total number of these paths.   

In order to do so, we are going to establish some bijections as it is standard in combinatorics. First, if we change the movement (1,1) to (0,1), there is a bijection with the paths starting at (0,0) and ending on $\{(x,m)~:~x\in\bN\}$. Second, if we add to these paths the movement (0,1) and the complete them with $(1,0)$ till they reach the diagonal, there is a bijection with the paths starting at $(0,0)$ and ending at $(m+1,m+1)$ without going below the diagonal. It is standard\footnote{For instance, this is the number of Dyck paths, see Problem 28, 52 and Theorem 1.5.1  in \cite{Sta15} for more details on the bijections. Also, these numbers appear in the ballot problem: suppose $A_1$ and $A_2$ are candidates for some election and there are an even number of voters, say $2n$. Let us also assume that $n$ voting for $A_1$ and $n$ for $A_2$. In how many ways can the ballots be counted so that $A_1$ is always ahead of or tied with $A_2$? See the aforementioned theorem.} that the total number of paths is $C_{m+1}$, where $C_n$ represents the $n$-th Catalan number, i.e.,
\begin{equation*}
	C_n\coloneqq \frac{1}{n+1}\binom{2n}{n}.
\end{equation*}
Now, note that we can express \eqref{eq:border growth cond} equivalently as,
\begin{equation*}
	\mu(\mathcal{B}_{n_0,n})=\mu(\mathcal{B}_{n_0,n_0})-\sum_{i=0}^{\frac{n-n_0}{2}-1}\beta_{n_0+2i}\gamma_{n_0,n_0+2i}
\end{equation*}
using the equations above \eqref{eq:border cond orig} where
\begin{equation*}
	\gamma_{n_0,n}\coloneqq\mu\left(\mathcal{B}_{n_0,n}\right)\alpha_{n_0,n}=\mu\left(\mathcal{B}^b_{n_0,n}\right).
\end{equation*}
This is useful because in our case of $n_0=1$ and $n=2m+1$ we can compute that sum easily,
\begin{equation*}
	\mu(\mathcal{B}_{n_0,n_0})-\sum_{i=0}^{\frac{n-n_0}{2}-1}\beta_{n_0+2i}\gamma_{n_0,n_0+2i}=\frac{1}{2}-\frac{1}{4} \sum _{i=0}^{m-1} \frac{C_{i+1}}{2^{2 i+1}}=2^{-2 (m+1)} \binom{2 (m+1)}{m+1}.
\end{equation*}
By Stirling's approximation,
\begin{equation*}
	\mu(\mathcal{B}_{1,2m+1})={\sqrt{\frac{1}{\pi m}}}+O\left(m^{-3/2}\right)\to 0
\end{equation*}
as $m\to\infty$. 
\section{Proof of Theorem \ref{thm:measure weighted CJT}}
\label{ap:proof weight}
	First note that the first hypotheses ensure that $w\in[1,W]$ as if $p$ is given, then $a(p)=1-w_d(p)$, $b(p)=W-w_d(p)$ and  $\varepsilon=w-w_d(p)$. Let us explore the first condition of Proposition \ref{prop:weighted CJT}:
	\begin{equation}
		w(p-q)=-(1+\varepsilon)+2(1+\varepsilon)p+(1-W) p^k+2 (W-1)p^{k+1}.
	\end{equation}
	We can analyze the expected values. As $p\in[0,1]$, then
	\begin{equation*}
		\bE\left(p^k\right)\ge\bE\left(p^{k+1}\right).
	\end{equation*}
	Nevertheless,
	\begin{equation}\label{eq:Holder}
		\bE\left(p^{k+1}\right)^\frac{1}{k+1}\ge\bE\left(p^{k}\right)^\frac{1}{k}
	\end{equation}
	by H\"older's inequality. 
	Let us show that there is a $k$ such that:
	\begin{equation*}
		2\bE\left(p^{k+1}\right)-\bE\left(p^{k}\right)>0.
	\end{equation*} 
	Indeed\footnote{We define $\nu_0\coloneqq \tilde{\nu}_0\circ\pi$ where $\pi(p,\varepsilon)=p$, i.e., the push-forward measure. Thus, $\tilde{\nu}_0(\{p\in A\})=\nu_0(A)$.},
	\begin{equation*}
		2m^{k+1}-m^k=\int_{[0,1/2)}x^k(2x-1)d\nu_0(x)+\int_{(1/2,1]}x^k(2x-1)d\nu_0(x).
	\end{equation*}
	Thus,
	\begin{equation*}
		2^k\left(2m^{k+1}-m^k\right)=\int_{[0,1/2)}2^kx^k(2x-1)d\nu_0(x)+\int_{(1/2,1]} 2^kx^k(2x-1)d\nu_0(x).
	\end{equation*}
	By the Dominated Convergence Theorem, the first term goes to zero as $k\to\infty$ and $2^kx^k(2x-1)\to\infty$ as $k\to\infty$ for the second term, so there is a $k$ such that the LHS is positive.
	
	Now we analyze the expectations that involve the error term $\varepsilon$, in particular,
	\begin{equation*}
		\bE\left((2p-1)\varepsilon\right)=\bE\left((2p-1)\bE\left(\varepsilon|p\right)\right)
	\end{equation*}
	by the law of iterated expectations. It is well-known that
	\begin{equation}\label{eq:mean truncated gaussian}
		\bE\left(\varepsilon|p\right)=\sigma {\frac {\phi (\alpha )-\phi (\beta )}{\Phi (\beta )-\Phi (\alpha )}},
	\end{equation}
	where $\alpha\coloneqq\frac{a}{\sigma}$, $\beta\coloneqq\frac{b}{\sigma}$, $\phi$ the p.d.f. of a standard Gaussian function and $\Phi$ its c.d.f.
	Then,
	\begin{equation*}
		\bE\left(\varepsilon|p\right)=(W-1)f(x,p),
	\end{equation*}
	where $x\coloneqq (W-1)/\sigma$ and
	\begin{equation*}
		f(x,p)\coloneqq \frac{\phi ((1-p) x)-\phi (-p x)}{x (\Phi (-p x)-\Phi ((1-p) x))}.
	\end{equation*}
	It is straightforward to check that for $p\in(0,1)$, $\bE\left(\varepsilon|p\right)\to 0$ exponentially as $x\to\infty$ because
	\begin{equation*}
		\alpha=\frac{(1-W)p}{\sigma}=-px\to -\infty,\quad \beta=\frac{(W-1)(1-p)}{\sigma}=(1-p)x\to\infty
	\end{equation*}
	as $x\to\infty$. If $p=0$, then $\beta$ still goes to infinity and if $p=1$, $\alpha$ still goes to $-\infty$.
	By the Dominated Convergence Theorem (as \eqref{eq:mean truncated gaussian} ensures the integrand is bounded by continuity on a compact set) we conclude that:
	\begin{equation*}
		\lim_{x\to\infty}\bE\left((2p-1)f(x,p)\right)=0.
	\end{equation*}
	Therefore,
	\begin{equation*}
		\frac{1}{n}\sum_{i=1}^n w_i(p_i-q_i)\to 2 m^1-1+(W-1)\bE\left(2 p^{k+1}-p^k+(2p-1)f(x,p)\right)~\text{a.s.}
	\end{equation*}
	as $n\to \infty$ by the SLLN. By the discussion above, if $x,W$ are large enough, then the limit is positive.
	
	For the denominator of the first condition in Proposition \ref{prop:weighted CJT} we know that
	\begin{equation*}
		\bE\left(w^2p(1-p)\right)>0
	\end{equation*}
	as $w\ge 1$, $p(1-p)\ge 0$ if $p\neq0,~p\neq1$ $\nu_0$-almost surely. The first case is rejected because $\delta_0((1/2,1])=0$ and if $\nu_0=\delta_1$, then the CJT holds for $W=1$ trivially (in this case we do not need $W$ large). Thus, by the SLLN again,
	\begin{equation*}
		\frac{\sum_{i=1}^n w_i(p_i-q_i)}{\sqrt{\sum_{i=1}^n w_i^2p_iq_i}}=\frac{\sqrt{n}\frac{1}{n}\sum_{i=1}^n w_i(p_i-q_i)}{\sqrt{\frac{1}{n}\sum_{i=1}^n w_i^2p_iq_i}}\to \infty.
	\end{equation*}
	Hence, the CJP now has $\mu_0$-measure equal to one, i.e., $\mu_0(\cC_I^w)=1$. As we did in the proof of Theorem \ref{th:unweig and indep}, the same holds for a ``deviation'' of this measure. Indeed, it follows from $\mu\ll\mu_0$ and the fact that the complement of $\cC_I^w$ is a $\mu_0$-null set.

\begin{remark}
	We could use weaker hypotheses, as in Theorem \ref{th:unweig and indep and biased}, nevertheless we opted for maintaining the simplicity. For instance, we could replace the independence of $\varepsilon$ by ergodicity and use the Ergodic Theorem instead of the SLLN, replace the Gaussianity by a nice enough distribution or make the parameters of the distribution depend on $p$. 
\end{remark}
\section{Practical implementation of epistemic weights: from psychology and political philosophy}\label{ap:psy and phil}
First, we must not confuse competence ($p$ near one) with other attributes that can be, under some conditions, correlated with competence, such as fluid or crystallized intelligence. Following \cite{Kah03} we can classify cognitive processes into two broad categories: System 1 (intuition) and System 2 (reasoning). The former is autonomous (executed automatically upon encountering the triggering stimulus and independent on input from high-level control systems). Furthermore, it is fast, emotional and relies on heuristics that can lead to biases. System 2 is slow, effortful, analytic... Many processes of System 1 can operate at once in parallel, but System 2 processing is largely serial. But we can split  System 2 further into two ``minds'', \cite{Sta09}, the algorithmic and reflective mind. The former deals with slow thinking and demanding computations (fluid intelligence, which IQ tests try to measure) and the latter is related to rational thinking dispositions and its functions are to \textit{initiate} the override biased responses of System 1, the ones based on a ``focal model" which can be biased or the simulation of alternative responses. Thus, rationality is a combination of both minds\footnote{Also, the autonomous mind or System 1 can provide rational responses as it might contain normative rules that have been tightly compiled and that are automatically activated as a result of overlearning and practice.}, not just the algorithmic one. Obviously, these systems need knowledge to work properly (and the one acquired through learning and past experiences is usually called crystallized intelligence), see Figure 3.3 of \cite{Sta09}. Nevertheless, note that some knowledge can be useless or harmful for achieving competence (``contaminated mindware'', \cite{Sta09}) or, even if necessary, remain unused, as in the ``override failure''. To be more specific, \cite{SW08}:
\begin{quote}
	...the relevant mindware for our present discussion is not just generic procedural knowledge, nor is it the hodge-podge of declarative knowledge that is often used to assess crystallized intelligence on ability tests. Instead, it is a very special subset of knowledge related to how one views probability and chance; whether one has the tools to think scientifically and the propensity to do so; the tendency to think logically; and knowledge of some special rules of formal reasoning
	and good argumentation.
\end{quote}

Thus, we must note that competence could not be achieved even if the algorithmic mind is ``highly developed''. There is evidence that thinking errors are relatively independent of cognitive abilities \cite{SW08}. For instance, there is not a significant correlation between the magnitude of some classical bias popularized by Kahneman \cite{Kah11} (e.g., anchoring effects or conjunction fallacy) and cognitive abilities. Another important example is the so-called myside bias (``people evaluate evidence, generate evidence, and test hypotheses in a manner biased toward their own prior opinions and attitudes''). The authors conjecture that fluid intelligence is only important when there is not a mindware gap (e.g., missing probability or scientific knowledge) and the need to override heuristic responses is detected. This is the case, e.g., in the rose syllogism (all flowers have petals; roses have petals; therefore, roses are flowers--which is invalid) and the belief bias, but not in the Linda problem between-subjects and conjunction fallacy. This feature of Linda problem illustrates an important fact; it is not enough to have the knowledge (here, basic probabilistic knowledge, $\bP(A\cup B)\ge \bP(A)$), but \textit{we must have the tendency to use it when needed}, specially when there are no cues to do so. Thinking dispositions, in contrast to cognitive abilities, are viewed as more malleable and this would predict that these skills are more teachable. As we were saying, there are some biases which are correlated with cognitive abilities as the argument evaluation test (\cite{SW08}), but they are not naturalistic or similar to a real-life situation because subjects have been told to decouple prior beliefs from the evaluation of evidence. Then the correlation happens because ``participants of differing cognitive abilities have different levels of computational power available for the override operations that make decoupling possible'', \cite{Sta13}.

Furthermore, more fluid intelligence could be even worse for myside bias. Indeed, in \cite{KPCS17} (see also references therein for more evidence) we can see why. In this experiment, subjects must draw valid causal inference from empirical data. The same empirical data is presented in two ways: in not an ideologically loaded way (skin-rash treatment) and as a partisan issue (gun-control ban). In the former (as expected), the higher the numeracy, the better the responses, but in the latter responses became polarized between liberal democrats and conservative republicans, less accurate and got worse for subjects with higher numeracy skills (algorithmic intelligence). Thus, this could be seen as a conflict between being epistemically rational (fitting one's beliefs to the real world, what is true) and instrumentally rational (optimizing goal fulfillment, what to do). This motivated reasoning can be the means to achieve our goals because sharing some political views is a symbol of membership and loyalty in political groups, expressive rationality, which can be more valuable than epistemic goals. In our day-to-day actions having true beliefs (epistemic rationality) is useful for achieving our goals (instrumental rationality). More precisely, if $\cA\coloneqq\{a_i\}_{i\in\cI}$ are the possible actions, $\cS\coloneqq\{s_j\}_{j\in\cJ}$ the possible states of the world and $\varphi:\cA\times\cS\to\cS$ maps the consequences of the action in each state of the world, then
\begin{equation*}
	U(a)=\sum_{j\in\cJ}\mathcal{P}(s_j\mid a)\cdot u\left(\varphi(s_j,a)\right), 
\end{equation*}
where $U$ is the von Neumann–Morgenstern utility function and $\mathcal{P}$ assigns probabilities to each state of the world. In order to maximize $U$, $\max_{a\in\cA}U(a)$ (instrumental rationality), we need to have correct beliefs about the world, $\cA$, $\cS$, $\mathcal{P}$ and $\varphi$, i.e., epistemic rationality. But in the political process our beliefs are dissociated from their consequences (one's beliefs on gun-control bans are unlikely to affect political decisions and their consequences), so expressive rationality makes perfect sense as epistemic and instrumental rationality are not necessarily linked and having true beliefs about the world could be less valuable than rejecting our previous beliefs or shared beliefs with our political group. As the social psychologist Jonathan Haidt puts it, we are good rationalizers but poor reasoners when thinking about politics. To achieve an epistemically rational response it could be more useful, for instance, to adopt measures that effectively shield decision-relevant science from the identity-protective motivated reasoning: behaving like a sport hooligan should not be seen an appropriate way to process information. In a recent (preregistered) replication of this study \cite{PAKVT21}, the effect of motivated reasoning was found but it was less clear the motivated numeracy (motivated reasoning increases with numeracy) finding. In another study, \cite{KS16}, they corroborate the same hypothesis of expressive rationality using beliefs about human evolution.

Hence, algorithmic intelligence might not be sufficient for rational thinking and not as necessary as one could initially think, for instance, if epistemically reliable shortcuts are available instead of a direct investigation or simulation of alternative responses. For example, if there is consensus between experts, take that as the most likely option. This could reduce the need for algorithmic intelligence but it does not eliminate some minimal amount; finding a reliable shortcut is a computation demanding process.  One should be cautious when assessing weights because they must be correlated with epistemic rationality and the relation between this and other typical measures of intelligence or knowledge is not trivial as we have seen. For instance, one proposal could be \cite{Sta16,StaRQ16} (total or partial subsets of the CART focused on epistemic rationality) in combination with particular knowledge (mindware) or skills (algorithmic mind) needed for competence in the particular domain of the choice we face. Any other metric that is correlated with this assessment or a similar one could be used too. 

Obviously the weight assignment will depend on the particular process under consideration. The assignment for a jury in a criminal trial will not be the same as the one for a democratic process (where part of the evidence presented above fits better). Nevertheless, the main idea still holds: a major part of the assignment should be based on epistemic rationality. But particular mindware should be considered in each situation. For instance, law and the particular criminal evidence for a trial and some basic knowledge of social sciences for a democratic process. Notice that some topics are more prone than others to be solved as epistemic rationality increases. For instance, discussing the means to achieve an agreed end can be easier than discussing the ends we should pursue.

Second, we could think that a more natural way to achieve \eqref{eq:CJT cond} is to exclude voters with $p<1/2$ (that is, $w=0$), which will imply $b>0$ and the CJP will hold almost surely (similar proof as Theorem \ref{th:unweig and indep} or \ref{th:unweig and indep and biased}). That is, as we said in the introduction we could consider:
\begin{itemize}
	\item $w_i=0$ if $p_i\le1/2$ (similar to expert rule) or,
	\item $w_i<0$ if $p_i<1/2$, as in \eqref{eq:optimal weights}.
\end{itemize}
Nevertheless, we opted to investigate the case of $w\ge 1$, i.e., all votes count (obviously, not in the same proportion) for several reasons. One is that it might be objected that in some circumstances not allowing some voters to participate can express disrespect, i.e., a semiotic objection based on the expressive value of the democratic process, \cite[Chapter 5]{Bre16}. To analyze it, the right of a competence decision process must be weighted against the somehow socially conceived expressive value of the restrictions. But in the setting where every potential voter is guaranteed a minimal weight, $w\ge 1$ for every voter, these objections are less motivated. In fact, votes have different weights in many present processes, although they are not usually weighted according to competence but other factors.
\bibliographystyle{siam}
{\bibliography{BibCJT}}
\end{document}